\documentclass[11pt]{iopart}

\usepackage{xcolor}
\usepackage{amssymb}
\usepackage{bm}

\usepackage{graphicx}
\usepackage{amsmath,amsthm,amsfonts,amscd}
\usepackage{amssymb}
\usepackage{url}
\usepackage{fullpage}
\usepackage{mathtools}
\usepackage[all]{xy}
\usepackage{slashed}
\usepackage{multirow}
\usepackage{hyperref}

\newtheorem{thm}{Theorem}[section]
\newtheorem*{thm*}{Theorem}
\newtheorem{cor}[thm]{Corollary}
\newtheorem{lemma}[thm]{Lemma}
\newtheorem{prop}[thm]{Proposition}

\newtheorem{defn}[thm]{Definition}

\theoremstyle{remark}
\newtheorem{remark}[thm]{Remark}

\newtheorem{example}[thm]{Example}
\newtheorem{remarks}[thm]{Remarks}

\newcommand{\wt}{\ensuremath{\widetilde}}

\newcommand{\R}{\ensuremath{\mathbb{R}}}
\newcommand{\N}{\ensuremath{\mathbb{N}}}
\newcommand{\Z}{\ensuremath{\mathbb{Z}}}

\newcommand{\C}{\ensuremath{\mathbb{C}}}
\newcommand{\T}{\ensuremath{\mathbb{T}}}
\newcommand{\GM}{\ensuremath{\mathbb{G}}}

\newcommand{\PMM}{\ensuremath{\mathbb{P}}}

\def\calT{\mathcal{T}}

\def\calL{\mathcal{L}}

\def\calB{\mathcal{B}}
\def\calH{\mathcal{H}}

\def\calF{\mathcal{F}}
\def\calA{\mathcal{A}}

\def\calW{\mathcal{W}}

\def\calG{\mathcal{G}}

\def\bP{\mathbf{P}}

\newcommand{\ol}{\overline}

\theoremstyle{definition}

\DeclareMathOperator{\Dom}{Dom}

\DeclareMathOperator{\Index}{Index}

\DeclareMathOperator{\Ker}{Ker}
\DeclareMathOperator{\coKer}{coKer}

\DeclareMathOperator{\Ran}{Ran}
\DeclareMathOperator*{\res}{res}

\newcommand{\sD}{\slashed{D}}

\begin{document}

\title{Non-Commutative Chern Numbers for \\
Generic Aperiodic Discrete Systems}

\author{Chris Bourne}

\address{Advanced Institute for Materials Research, 
\\Tohoku University, 
\\Sendai, 980-8577, Japan \\
\href{mailto:chris.bourne@tohoku.ac.jp}{chris.bourne@tohoku.ac.jp}}

\author{Emil Prodan}

\address{Department of Physics and
\\ Department of Mathematical Sciences 
\\Yeshiva University, 
\\New York, NY 10016, USA \\
\href{mailto:prodan@yu.edu}{prodan@yu.edu}}

\date{\today}

\begin{abstract} 
The search for strong topological phases in generic aperiodic materials and meta-materials is now vigorously pursued by the condensed matter physics community. In this work, we first introduce the concept of patterned resonators as a unifying theoretical framework for topological electronic, photonic, phononic etc. (aperiodic) systems. We then discuss, in physical terms, the philosophy behind an operator theoretic analysis used to systematize such systems. A model calculation of the Hall conductance of a 2-dimensional 
amorphous lattice is given, where we present numerical evidence of its quantization in the mobility gap regime. Motivated by such facts, we then present the main result of our work, which is the extension of the Chern number formulas to Hamiltonians associated to lattices without a canonical labeling of the sites, together with index theorems that assure the quantization and stability of these Chern 
numbers in the mobility gap regime. Our results cover a broad range of applications, in particular, those involving quasi-crystalline, amorphous as well as synthetic ({\it i.e.} algorithmically generated) lattices. 
\end{abstract}

\maketitle

\setcounter{tocdepth}{2}
{\scriptsize
\tableofcontents
}

\section{Introduction}

Topological insulators \cite{Hal1988,KaneMele2005I,KaneMele2005II,BHZ2006,KWB2007,MooreBalents2007,FuKane2007,HQW2008} 
have attracted intense  interest from the condensed matter community. By definition, these are crystalline solids whose electronic degrees 
of freedom display quantized bulk and surface responses to external stimuli, even in the regime of strong disorder. 
For example, they were theoretically predicted to retain their topological characteristics up to the room temperature, 
though this remains to be demonstrated experimentally. The thermodynamic data for the classical atomic degrees 
of freedom of a disordered crystal, which in our view is just a crystal at finite 
temperature, is encoded in a dynamical system $(\Omega,\GM,{\rm d}\PMM)$, where $\Omega$ is the configuration space of the atomic 
degrees of freedom, $\GM$ is the space group of the crystals acting on $\Omega$ and ${\rm d} \PMM$ is the Gibbs measure for the atomic 
degrees of freedom, defined over $\Omega$ \cite{KuhneProdan2017}.  For the thermodynamically pure homogeneous phases usually 
studied in laboratories, the Gibbs measure must be invariant and ergodic with respect to the $\GM$-action \cite{Dobrushin1989,Ruelle1969}. 

\vspace{0.2cm} 

The quantum dynamics of the electron degrees of freedom is generated by a covariant family of Hamiltonians with respect to 
$(\Omega,\GM,{\rm d}\PMM)$. Such covariant families of observables can be described quite generally using 
representations of a crossed product algebra  \cite{Bellissard1986}. For discrete systems and ignoring point symmetries, 
the crossed product is 
simply by $\Z^d$ ($d=$ physical space dimension) and,  as such, the mathematical structure of the topological phases supported 
by disordered crystals is the simplest among the condensed matter systems. For this reason, such disordered crystals  
are quite well understood at 
this time. Indeed, the condensed matter physics community put forward a conjecture in the form of a classification table of all possible 
disordered crystalline phases displaying metallic electron transport at the boundaries of the 
samples \cite{SchnyderPRB2008qy,Kitaev2009hf,RyuNJP2010tq}. The conjecture survived a large number of numerical tests 
and, at the rigorous level, good progress towards a proof has been achieved in quite a large number 
of works. We remark that crystalline solids recently attracted a renewed interest due to the existence of 
topological phases that are solely stabilized by the point symmetries of the crystals \cite{SSG17,PVW2017,BradlynNature, KdBvWKS, SSG18}.

\vspace{0.2cm} 

Inspired by the research on topological insulators, similar effects are now also sought in photonic
 \cite{RZP2013,WCJ2009,HMF} and phononic \cite{PP2009,KL2013,PCV2015,NKR2015} crystals, 
 as well as plasmonic \cite{SR2016} systems. These are much more versatile platforms that enabled 
 experimentalists to look beyond the periodic table and investigate almost-periodic \cite{KLR2012,KRZ2013,Pro2015,HPW2015}, 
 quasi-crystalline \cite{KZ2012,VZK2013,TGB2014,TDG2015,VZL2015,LBF2015,DLA2016,BRS2016,BLL2017,FuchsVidal} 
 and even amorphous patterns \cite{MNH2016,AS17}. While many of these models can be treated 
 within the framework of (discrete) crossed product algebras,  see {\it e.g.} \cite{Pro2015,HPS2017}, 
 amorphous patterns can not in general. 
The difference comes from whether there is a canonical labeling of the lattice by $\Z^d$ 
such that the Hamiltonians, which are defined on the same physical Hilbert space, 
remain short range with respect to these $\Z^d$-labels. For amorphous 
patters, this can not be done. For quasi-crystalline patterns, we can reduce our system to 
a short-range $\Z^d$-labelling using the results of Sadun and Williams~\cite{SW03}, but at the 
expense of altering the underlying lattice. If possible, we would like to avoid this step.

\vspace{0.2cm}

The index theorems developed for strongly disordered 
 crystals \cite{BES1994,PLB2013,PS2016} are specialized for crossed product algebras and no 
 longer work in the amorphous setting (or quasi-crystalline lattices without alterations). 
 Hence, a gap emerged in our understanding of the novel topological phases. 
 Indeed, the pioneering works \cite{MNH2016,AS17} on topological amorphous phases brought great excitement,
 but also raised a number of fundamental questions that still puzzle the condensed matter community. While the topological 
 amorphous phase in \cite{MNH2016} was realized in the laboratory with classical mechanical systems, we will 
 model an analogous system using a two-dimensional homogeneous amorphous crystal under a uniform 
 perpendicular magnetic field. The fundamental questions, however, remain the same.
 
\vspace{0.2cm}  
 
 To streamline the discussion we introduce two terminologies, the thermodynamic phase and the topological phase, 
 where the former refers to the classical atomic degrees of freedom while the latter refers to the quantum degrees of 
 freedom of the electrons. This is by no means a standard terminology. If the amorphous thermodynamic phase is pure, 
 the macroscopic transport coefficients are well defined, {\it i.e.} the experimentally measured direct and Hall conductivities 
 $\sigma$ and $\sigma_H$, respectively, have fixed values that do not fluctuate from sample to sample, even though the 
 atomic configurations can be vastly different. Let us point out that while spectral gaps 
 may occur in amorphous systems, for the models considered in~\cite{MNH2016}, 
 mobility gaps are more common, where the direct conductivity vanishes asymptotically 
 as temperature $T$ is lowered towards zero (see Section~\ref{SubSec-NemResults}). Now, suppose the Fermi 
 level $E_F$ (or better said chemical potential) is located in one of the mobility gaps. We present some outstanding questions:
\begin{enumerate}
\item Does $\sigma_H$ have a limit as $ T\searrow 0$? 
\item If yes, is there a formula for $\sigma_H$ akin to the (non-commutative) Chern number?
\item Is $\sigma_H$ quantized as in the case of strongly disordered crystals?
\item Is the amorphous Hall phase the same as the one observed in disordered crystals?
\end{enumerate}

\vspace{0.2cm} 

Questions (i-ii) relate to the physical interpretation of the topological invariant, which was one of the central points of 
discussion in \cite{MNH2016}. There, the authors tried to adapt a formula due to Kitaev \cite{Kitaev2006}, but 
questions (i-ii) already had affirmative answers provided by the general theory of electron transport in homogeneous 
systems developed in \cite{BES1994} (see also \cite{SBBI1998,SBBII1998}). In these works, one can find the equivalent 
of the so called TKNN formula \cite{TKKN1982}, derived similarly from the zero temperature limit of the Kubo--Green formula, 
this time in the context of homogeneous (as opposed to periodic) systems. Up to a physical constant, it takes the form:
\begin{equation}\label{Eq-HallCond}
\lim_{T \searrow 0} \sigma_H = \Tr_\mathrm{Vol} \Big \{ P_F \big [ [X_1,P_F],[X_2,P_F] \big ] \Big \},
\end{equation}
where $\Tr_\mathrm{Vol}$ represents the trace per volume, $X$ is the position operator and $P_F$ is the Fermi projector, 
{\it i.e.} the spectral projector of the Hamiltonian on $(-\infty,E_F]$. Also, $[\cdot,\cdot]$ stands for the commutator of 
two operators. The relation between \eqref{Eq-HallCond} and Kitaev's formula used in \cite{MNH2016} is not understood 
at this time. Ref.~\cite{AS17} uses a dated version of the Bott index \cite{LH2011} \footnote{ An alternative version 
of the Bott index defined in \cite{LH2011} appeared in \cite{LS2017, LSEven} and these 
new versions are connected to the Fredholm indices appearing in our work. As such, our local index formulas 
connect them to \eqref{Eq-HallCond}.} for which there is no local formula hence no relation to \eqref{Eq-HallCond} 
can be established.

\vspace{0.2cm} 

For an amorphous solid, there was no a priori reason, up to now, to believe that \eqref{Eq-HallCond} 
remains quantized in both the spectral and mobility gap regimes, as $\sigma_H$ could 
very well behave like a weak topological invariant in this new setting. We recall that, for a disordered crystal, the stability 
and quantization of $\sigma_H$ in the mobility gap regime follows from the index theorem derived in \cite{BES1994}. 
As we already mentioned, this fundamental result is highly specialized to the context of disordered crystals and should 
not be generalized beyond that. Without such index theorem for amorphous solids, to tell us the precise conditions in 
which $\sigma_H$ is stable and quantized, there is no way to answer questions (iii-iv) from above.

\vspace{0.2cm}

The main results of our work are index formulas for \eqref{Eq-HallCond} and its higher dimensional generalizations, 
as well as for the odd-dimensional versions. The formulas apply to generic homogeneous Hamiltonians over (Delone) 
point-patterns, in particular, to amorphous solids. They are formulated in Theorems~\ref{thm:complex_bulk_pairing_even} 
and \ref{thm:complex_bulk_pairing_odd} for the spectral gap regime and in Section~\ref{subsec:mobility_gap_index} for the mobility 
gap regime. As we shall see, the stability and quantization of the topological invariants require the Gibbs measure of the 
atomic degrees of freedom to be ergodic with respect to the  space translations. As such, the Hall plateaus can be observed 
only in pure thermodynamic phases, which, from a physical point of view, makes perfect sense because, otherwise, 
the trace per volume in \eqref{Eq-HallCond} will depend on how one achieves the thermodynamic limit 
({\it i.e.} on the boundary conditions). 

\vspace{0.2cm} 

We now can answer questions (iii-iv). If we assume that the amorphous and disordered 
crystalline solids are distinct pure thermodynamic phases, which in general is the case, then ergodicity of the Gibbs measure 
is necessarily lost while trying to deform these systems into each other. As such, the Hall conductance can change 
its quantized value during the deformation, even if the mobility gap stays open. This leads us to the following conclusions:
\begin{enumerate} 
\item Crystalline, amorphous and many other pure thermodynamic phases can host topological phases of the electronic degrees of freedom.
\item The work \cite{MNH2016} showed for the first time a topological phase hosted by a thermodynamic pure phase other than a crystal. 
Without doubt, it is a new state of matter.
\item At the phase boundaries between pure thermodynamic phases, the topological phases might not be aligned, 
{\it i.e.} the topological numbers can change as this border is crossed.
\end{enumerate}

Our main technical tool we use to prove quantization is the (unbounded) index theory of the 
$C^*$-algebra associated to the transversal groupoid of a point pattern. This groupoid and 
algebra was first considered by Bellissard in~\cite{Bellissard1986} and further developed by Kellendonk 
to study the dynamics of tilings and applications to 
the gap labelling conjecture~\cite{Kellendonk95, Kellendonk97}. Algebraic, homological and 
spectral properties of this groupoid and its $C^*$-algebra have been studied quite extensively, 
see for example~\cite{BHZ00, BBG06, LPV07, BelSav}. We note that
groupoid $C^*$-algebras and their associated index theory also 
played a role in the description of the quantum Hall effect on the hyperbolic plane~\cite{CHMM, CHM99} 
and coarse-geometric descriptions of topological phases~\cite{Kubota15b}. Let us also point out that in the spectral gap regime, 
the Fredholm indices involved in our work 
can be exactly computed on finite volumes using the methods developed in \cite{LS2017, LSEven}. 

\vspace{0.2cm} 

In this paper we construct a spectral triple for the transversal groupoid $C^*$-algebra 
which satisfies the hypothesis of the local index theorem in non-commutative 
geometry~\cite{CoM,CPRS2,CPRS3}. We can then compute the index formula, which 
recovers the familar non-commutative Chern number formulas and automatically 
represents the $\Z$-valued analytic pairing of $K$-theory with the `Dirac operator' on the 
groupoid. While the algebra is different to the crossed product description, the computation 
of the index formula is very similar to previous 
studies~\cite{BRCont, BSBWeak}. We then extend this index formula to a larger 
Sobolev algebra with a characterisation similar to~\cite{PSBbook}. 

\vspace{0.2cm}

While many of our 
index theoretic results extend to the aperiodic/amorphous picture quite naturally, the 
connection of elements in the Sobolev algebra to observables with spectral regions 
of dynamical localisation is not as well established. A key technical hurdle is that 
we do not work with random Hamiltonians on a single lattice $\calL\subset \R^d$, but a 
family of lattices indexed by some configuration space $\{\calL\}_{\calL\in\Xi}$ and 
with different Hilbert spaces $\{\ell^2(\calL)\}_{\calL\in\Xi}$. A full investigation 
of the spectral properties of such operators, while desirable, is beyond the scope of this 
paper and we will instead focus on the index-theoretic aspects and their applications.

 \section{Patterned resonators} 
 \label{Sec:Patterned_resonators}

As we mentioned in our introduction, the interest in topological effects is rapidly broadening to meta-materials which 
enable controlled design of photonic, acoustic and plasmonic systems. Below, we introduce a simple overarching 
physical framework which puts all these systems on equal footing. Hence we can cover them all with same mathematical analysis.

\subsection{Definitions, examples, dynamics}

In our language, a resonator is a physical system confined to a small region of the physical space and having an 
arbitrarily large but nevertheless finite number of degrees of freedom. From the mathematical point of view, the 
resonator is a point with an internal structure. Attached to it, there are physical observables and a non-dissipative 
dynamics, which all can be described by linear operators over a finite dimensional Hilbert space that, of course, 
can be chosen to be $\C^N$. The number $N$ will be referred to as the number of internal degrees of freedom 
and $\C^N$ as the internal Hilbert space. Resonators will be represented schematically as in Fig.~\ref{Fig-CoupledResonators}. 
Below, we provide some examples for reader's convenience.

\begin{figure}
\center
\includegraphics[width=0.3\textwidth]{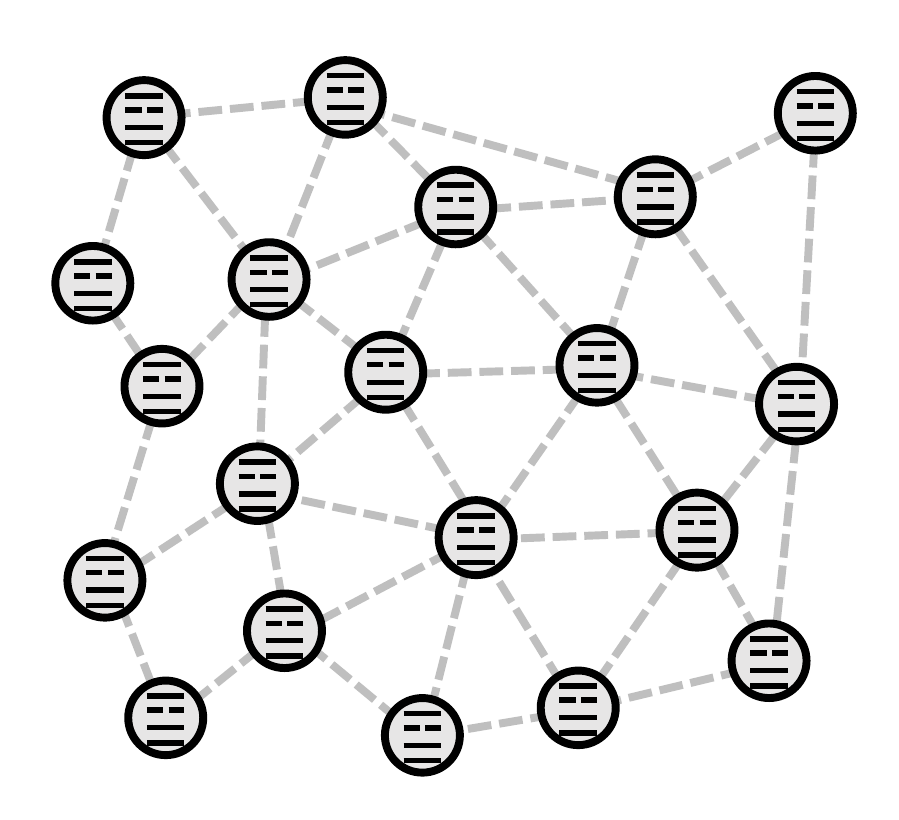}\\
  \caption{\small Schematic representation of resonators and their coupling when arranged in point pattern.
}
 \label{Fig-CoupledResonators}
\end{figure}

\begin{example}{\rm A confined quantum mechanical system with a finite number of quantum states is the 
prototype of the resonator. The atoms and molecules in an extended condensed matter system are often 
treated this way without losing the precision of the calculations. 
}$\, \Diamond$
\end{example}

\begin{example}{\rm A mechanical harmonic oscillator with $N$-degrees of freedom also fits our definition 
of a resonator. Indeed, if $(q_j,p_j)$, $j=1,\ldots,N$ are the generalized coordinates and the associated 
canonical momenta, then, by passing to the complex coordinates:
\begin{equation}
(q_j,p_j) \rightarrow \xi_j = \tfrac{1}{\sqrt{2}}(q_j+i p_j), \quad j=1,\ldots,N,
\end{equation}
Hamilton's equations take the form:
\begin{equation}
i \frac{{\rm d}\xi_j}{{\rm d} t} =\frac{\partial H}{\partial \xi_j^\ast}, \quad j=1,\ldots,N.
\end{equation}
A harmonic oscillator is defined by a quadratic Hamiltonian of the form:
\begin{equation}
H(\xi_1,\xi_1^\ast,\ldots,\xi_N,\xi_N^\ast) = \sum_{i,j=1}^N h_{ij}\, \xi_i^\ast \xi_j, \quad h_{ij}^\ast = h_{ji},
\end{equation}
hence Hamilton's equations reduce to:
\begin{equation}
i \frac{{\rm d}\psi}{{\rm d} t} = h \psi, \quad \psi=\begin{pmatrix} \xi_1 \\ \ldots \\ \xi_N \end{pmatrix} \in \C^N,
\end{equation}
where $h$ is the $N \times N$ matrix with the entries $h_{ij}$.
}$\, \Diamond$
\end{example}

\begin{example}{\rm The dynamical Maxwell equations without sources can be cast in the form of a linear 
Schr\"{o}dinger equation \cite{DL2014, DL2017}. Then the discrete electromagnetic resonant modes inside a cavity 
with reflecting walls provide additional examples of resonators, provided the higher frequency modes can be neglected. 
}$\, \Diamond$
\end{example}

\vspace{0.2cm}

When two or more resonators are brought close to each other, the dynamics of the internal modes couple 
due to either an weak overlap of the resonant modes or because the force fields or potentials extend far 
beyond the confining space of the resonators. The experimental signature of such a coupling, which in most 
cases can be mapped with great precision, is the hybridization of the resonant modes accompanied by shifts 
of the eigen-frequencies. In the regime of weak coupling and in the quadratic or single-electron approximations, 
the internal spaces remain unaltered and the dynamics of the coupled resonators takes place inside the Hilbert space 
\begin{equation}
\calH = \C^N \otimes \ell^2(\calL).
\end{equation}
The dynamics is then generated by a bounded Hamiltonian of the type:
\begin{equation}\label{Eq-GenHamiltonian}
H_\calL = \sum_{x,x'\in \calL} h_{x,x'}(\calL) \otimes |x \rangle \langle x' |, \quad h_{x,x'} \in M_N(\C), \quad h_{x',x}=h_{x,x'}^\ast, 
\end{equation}
where $\calL$ is the point pattern formed by the resonators, which for simplicity are considered all the same. 
Throughout, $M_N(\C)$ will denote the algebra of $N \times N$ matrices with complex entries. 

\begin{remark}\label{Re-HDep} We have used a notation that suggests that the 
hopping matrices $h_{x,x'}(\calL)$ depend not just on the points $x$ and $x'$ but on the entire pattern 
$\calL$. A subtle point which we want to stress is that the Hamiltonian is fully determined 
by the pattern but, of course, there is potentially a large amount of geometrical data encoded in $\calL$. 
$\, \Diamond$
\end{remark}

\begin{remark}\label{Re-HCont}
On the physical grounds, we can be sure that the hopping matrices $h_{x,x'}(\calL)$ depend continuously 
on $\calL$ (in a sense made precise later) and that they become less significant as the distance between 
$x$ and $x'$ increases.
$\, \Diamond$
\end{remark}

\begin{example} If $N=1$ and the individual resonant modes are isotropic, as well as the coupling occurs through 
the overlap of the exponentially decaying tails of these modes, then the Hamiltonian takes a universal form:
\begin{equation}\label{Eq-ModelHam1}
 H_\calL = \sum_{x,x'\in \calL} e^{-\beta|x-x'|} \, |x \rangle \langle x' |,
\end{equation}
in some adjusted energy units.
$\, \Diamond$
\end{example}

\begin{remark}If we adjust the length unit such that $\beta=1$ in the above example, then the 
hopping coefficients become less than $10^{-3}$ if $|x -x'| >7$ and, in many instances, they can 
be neglected entirely beyond this limit. When this is the case, the Hamiltonians are said to be of 
finite hopping range.
$\, \Diamond$
\end{remark}

\subsection{The structure of the Hamiltonians}
\label{SubSec-Standard}

Let us point out that, apart from the fact that the hopping coefficients are fully specified by the pattern $\calL$, 
the Hamiltonian in \eqref{Eq-GenHamiltonian} takes the most general form of a bounded operator over 
$\C^N \otimes \ell^2(\calL)$. Yet, as we shall see below, the Hamiltonians do have a certain structure and 
this is why they generate a subalgebra inside $\calB(\cal H)$, the algebra of bounded operators over $\calH$. 
Indeed, if the pattern is moved rigidly by some $y\in \R^d$, then consistency enforces a relation between the 
hopping coefficients of $H_\calL$ and $H_{\calL-y}$:
\begin{equation}
h_{x,x'}(\calL) = h_{x-y,x'-y}(\calL - y), \quad x,x' \in \calL.
\end{equation}
Then
\begin{equation}
H_\calL = \sum_{x,x'\in \calL} h_{x,x'}(\calL) \otimes |x \rangle \langle x' | = \sum_{x,x'\in \calL} h_{0,x'-x}(\calL-x) \otimes |x \rangle \langle x' |,
\end{equation}
and, if we introduce $q=x'-x \in \calL -x$, then:
\begin{equation}
H_\calL = \sum_{x\in \calL} \sum_{q\in \calL-x} h_{0,q}(\calL-x) \otimes |x \rangle \langle x+q |.
\end{equation}
As one can see, we can drop one subscript and write $h_q$ instead of $h_{0,q}$. Note that inside the sum, 
$x \in \calL$ as well as $x \in \calL-q$, for any $q \in \calL-x$. As such, $|x\rangle$ can be seen as a vector 
in $\ell^2(\calL)$ or in $\ell^2(\calL- q)$. Then, if we use the shift operators defined by the isometries:
\begin{equation}
S_q : \ell^2(\calL) \rightarrow \ell^2(\calL-q), \quad S_q|x\rangle = |x -q \rangle, \quad S_q^\ast |x-q\rangle = |x\rangle,
\end{equation}
the generic Hamiltonians \eqref{Eq-GenHamiltonian} start to display a very particular structure:
\begin{equation}\label{Eq-StrHamiltonian}
H_\calL= \sum_{x\in \calL} \sum_{q\in \calL-x} h_{q}(\calL-x) \otimes |x \rangle \langle x |S_q,
\end{equation}
where $|x\rangle \langle x|$ is understood as a partial isometry from $\ell^2(\calL-q)$ to $\ell^2(\calL)$. 

\vspace{0.2cm}

Let us point out a few remarkable facts about \eqref{Eq-StrHamiltonian}. First, the structure revealed itself 
because we consistently viewed the hopping matrices as functions over the space of patterns. Perhaps the 
significance of our Remark \ref{Re-HDep} becomes more clear now. Equally important is Remark~\ref{Re-HCont}, 
which tells that this functions are continuous and that in practice there is only a finite number of summations 
over $q$ in \eqref{Eq-StrHamiltonian}. Now, given one pattern $\calL$, we can always choose the origin of 
$\R^d$ such that one point of $\calL$ is positioned at the origin, or shortly $0 \in \calL$. Then notice 
in \eqref{Eq-StrHamiltonian} that only the patterns $\calL-x$ with $x \in \calL$ appear. These are all the rigid 
translates of $\calL$ with the property that $0$ is among their points. The set of these patterns will be denoted 
by $\Xi$ and will later be endowed with a topology. The important conclusion of our discussion is that in order to 
reproduce \eqref{Eq-StrHamiltonian}, we only need the values of the hopping matrices over the space $\Xi$. We hope 
that this convinces the reader that the algebra generated by all $H_\calL$'s, called the algebra of physical observables, 
is much smaller than $\calB(\calH)$. In fact, if the space $\Xi$ is simple enough, there are good chances that the 
$K$-theories of this algebra, which classify the gapped Hamiltonians over $\calL$, can be fully resolved.

\vspace{0.2cm}

The simplest example is that of a periodic pattern in which case $\Xi$ reduces to a point and the Hilbert spaces of 
the translates coincide. Then the algebra of observables is generated by $d$ commuting shift operators. If the 
pattern is not periodic but its points can be labeled by $\Z^d$ such that the translations $\calL-x$ reduce to the trivial 
action of $\Z^d$ onto itself, then the Hilbert spaces of the translates can be canonically identified and the algebra of 
observables turns out to be the crossed product $C(\Xi) \rtimes \Z^d$ with the obvious action of $\Z^d$. As we 
already mentioned in the introduction, here we are interested in the generic cases where the labeling by $\Z^d$ is not 
possible. In this case, the algebra of physical observables is the groupoid algebra introduced 
in~\cite{Bellissard1986, Kellendonk95} and discussed in Section~\ref{Sec-GrupoidAlg}.

\section{Quantization of Hall conductance in amorphous solids: Numerical evidence}
\label{Sec-QHall}

In this section we employ the numerical techniques developed in \cite{ProdanAMRX2013,ProdanSpringer2017} and 
perform numerical simulations of the Hall conductance \eqref{Eq-HallCond} for amorphous solids in dimension 2. 
This choice has been made precisely because these systems are quite different from disordered crystals. In 
particular, the labeling by $\Z^d$ discussed in the conclusions of Section~\ref{SubSec-Standard} does not exist. 
The numerical techniques from \cite{ProdanAMRX2013,ProdanSpringer2017} were developed for disordered crystals and  
hence have to be adapted to the new context. This is explained in Section~\ref{SubSec-NemResults}, though 
without any estimates of the numerical errors. 

\vspace{0.2cm}

An important issue is the requirement of a gap in the energy spectrum, which can be either spectral or dynamical. 
Existence of gaps in the spectrum of an amorphous solid is possible, but is generally not expected unless the couplings are strongly 
dependent on many geometrical data encoded in $\calL$ (see \cite{MNH2016}). In our work, however, we want to 
work with an isotropic coupling as in \eqref{Eq-ModelHam1}, but we will introduce a uniform magnetic field perpendicular 
to the sample. One remarkable observation is the opening of several large mobility gaps in the energy spectrum, 
which reminds us of the splitting of the continuous energy bands of electrons on a periodic lattice when subjected to 
a magnetic field. Let us point out that the integer quantum Hall effect has been always simulated using randomly 
perturbed periodic lattices, but the potential in the quantum wells where the effect is experimentally observed is in 
fact closer to that of an amorphous system. Hence, our theoretical and numerical results may lead to a better 
qualitative and quantitative understanding of this effect. 

\subsection{The system defined}

We describe first how the amorphous pattern was generated in our simulations. Firstly, we fixed the number of 
points per area, hence the density of points, and we chose the unit of length such that the fixed density 
becomes one point per unit square. We then produced a pattern $\calL_L$ of $N=L\times L$ sites on a 
flat 2-torus of equal circumferences $L \in \N$ (called the $L$-torus from now on), using the following algorithm:
\begin{itemize}
\item A random number generator was used to produce a new random point inside the square $[0,L]\times [0,L]$. 
Note that we include the boundaries.
\item The distances from this point to all already existing points were evaluated. 
The standard distance of the flat torus was used, hence periodic boundary conditions were automatically enforced.
\item If any of those distances were smaller than a predefined minimum distance $d_{\rm min} \leq 1$, 
then the newly generated point was rejected. Otherwise, the point was kept.
\item The cycle was repeated until all $N$ points were laid down on the flat torus.
\item The distances between all pairs of points was computed and if they were all found to be larger than a pre-defined $d_{\rm max}>d_{\rm min}$, the pattern was rejected. Otherwise, it was kept. 
\end{itemize}
The only input for the algorithm is the triple $(L,d_{\rm min},d_{\rm max})$, with the understanding that 
always $N=L\times L$. In all our simulations, $d_{\rm min}$ was fixed at $0.83$ while $L$ was varied 
from 60 to 120. One may note that the point pattern can be thought as the 
centers of a system of $N$ hard balls of diameter $d_{\rm min}$ dropped at random on the flat 2-torus of size 
$L \times L$. There is a small but nevertheless finite probability for the balls to cluster in large pockets, in which 
case large holes will emerge in our patterns. The last condition of our algorithm prevents this phenomena 
and keeps the patterns $d_{\rm max}$-relatively dense (see Definition \ref{defn-delone1}). In our simulation, though, 
we never observe this clustering phenomena.  

\begin{figure}
\center
\includegraphics[width=0.8\textwidth]{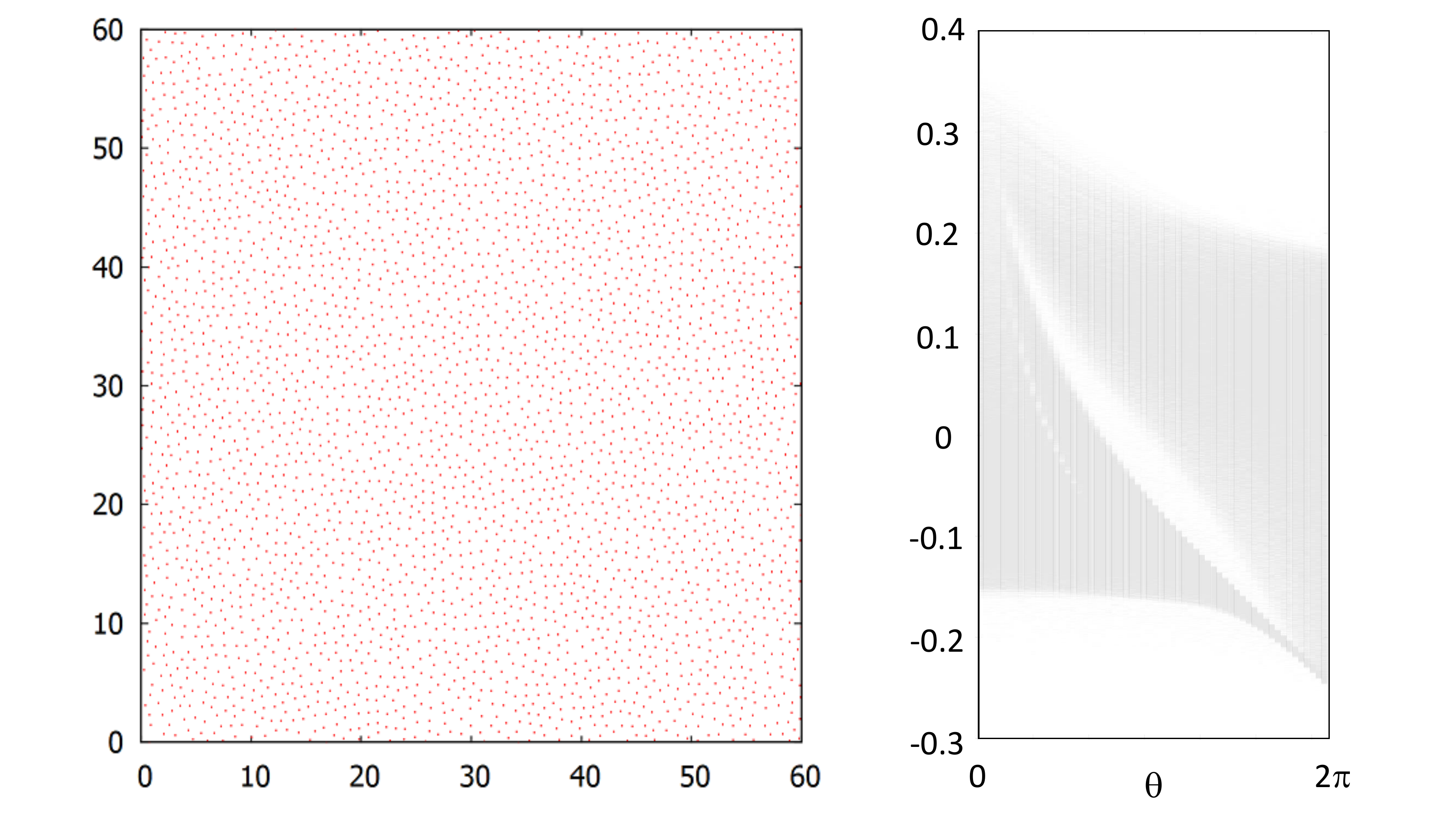}\\
  \caption{\small (left) Example of an amorphous pattern obtained with the algorithm described in the text. 
  The parameters are $L=60$ and $d_{\rm min}=0.83$; (right) The spectrum of the Hamiltonian \eqref{eq-modelH2} 
  as function of the strength of the magnetic field. The computation has been carried for a pattern with $L=120$ 
  and $d_{\rm min}=0.83$.
}
 \label{Fig-PattAndSpec}
\end{figure}

\vspace{0.2cm}

An example of a pattern generated with the above algorithm is shown in Fig.~\ref{Fig-PattAndSpec} for $L=60$. 
Note that the patterns are indeed periodic in the sense that, if the square $[0,L] \times [0,L]$ is wrapped in a torus, 
one will be unable to detect where the edges of the square were. For the same reason, we can periodically 
extend the pattern over the whole $\R^2$ without violating the constraints. We mentioned this detail because 
it relates to the program of finding periodic approximates of a pattern \cite{BeckusPhDThesis, BBDN17}. In our context, 
the patterns we are interested in are actually defined by the thermodynamic limit of the periodic ones. More 
precisely, note that every time the algorithm is run for a fixed $L$, the pattern will be different from the previous. 
Hence we are dealing with a family of patterns which can be periodically extended over the whole space. Let 
$\Omega_L$ be the set of $L$-periodic patterns that we can generate with our algorithm, which we close in 
the standard topology of the space of Delone sets (see Proposition~\ref{prop-delonespace}). The $\Omega_L$ 
spaces are invariant with respect to the translations of $\R^2$ and they form an inductive tower of compact topological spaces:
\begin{equation}
\Omega_{L} \subset \Omega_{2L} \cdots \subset \Omega_{2^n L} \subset \cdots \; .
\end{equation}
The ``configuration'' space of the infinite patterns is
\begin{equation}
\Omega = \overline{ \bigcup_{n \in \N} \Omega_{2^n L}},
\end{equation}
which is a compact space, invariant to the translations $T$ of $\R^2$. The pair $(\Omega,\R^2,T)$ is then a 
topological dynamical system, which also comes equipped with an invariant probability measure. 
Indeed, the finite volume algorithm determines entirely the probability of a $\mathcal L_L \in \Omega_L$ pattern to occur. 
Note that rigidly shifted patterns on the $L$-torus occur with equal probabilities. 
We denote by $\bP_L$ the associated finite-volume probability measure, which is invariant 
to the cyclic shifts of the $L$-torus. Then the measure $\bP$ on $\Omega$ can be defined as the unique 
measure whose traces over $\bar \Omega_{2^n L}$ coincide with $\bP_{2^n L}$ for all $n=1,2,\ldots$, where
\begin{equation}
\bar \Omega_{L}= \{\mathcal L \in \Omega, \ \mathcal L \cap [-\tfrac{1}{2}L,\tfrac{1}{2}L)^2 = 
\mathcal L_L \cap [-\tfrac{1}{2}L,\tfrac{1}{2}L)^2 \ \mbox{for some} \ \mathcal L_L\in \Omega_L \}.
\end{equation} 
One important issue is whether the measure $\bP$ is ergodic, which at this point we 
must assume.

\vspace{0.2cm}

The model Hamiltonian used in our simulations is:
\begin{equation}\label{eq-modelH2}
H_\calL: \ell^2(\calL) \rightarrow \ell^2(\calL), \quad H_\calL 
  =\sum_{x,x'\in \calL} e^{\imath \theta \, x\wedge x'} e^{-3|x-x'|}\, |x \rangle \langle x' |.
\end{equation}
Here, $e^{\imath \theta \, x \wedge x'}$ is the usual Peierls phase factor \cite{PeierlsZP1933} encoding the 
presence of a magnetic field, with $x \wedge x'=\tfrac{1}{2}(x_1x'_2-x_2 x'_1)$ being the oriented area of the 
triangle made out of $x$, $x'$ and the origin, and $\theta$ is the strength of the magnetic field in some 
adjusted units. Note that no cutoff was introduced on the hopping range.

\begin{figure}
\center
\includegraphics[width=\textwidth]{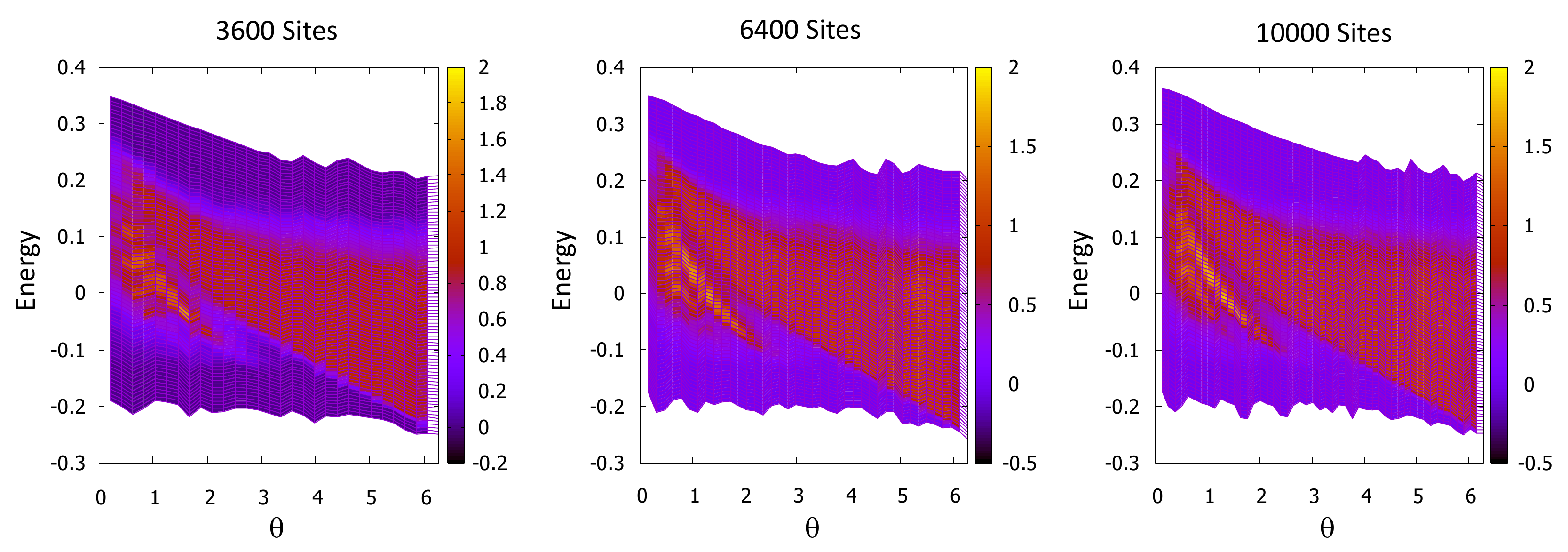}\\
  \caption{\small Map of the Hall conductance as function of Fermi energy and strength of magnetic field 
  for pattern parameters $d_{\rm min}=0.83$ and (left) $L=60$, (middle) $L=80$, (right) $L=100$. 
}
 \label{Fig-ChernButterfly}
\end{figure}

\subsection{Numerical implementation and results}
\label{SubSec-NemResults}

We selected from the configuration space $\Omega$ an $L$-periodic pattern $\calL$ and adapted 
the Hamiltonian \eqref{eq-modelH2} on the $L$-torus. To comply with the periodic boundary conditions, 
the values of the magnetic field were restricted to the discrete values 
$\theta_n=\frac{4\pi n}{L}$, $n=0,1,\ldots$. The energy spectrum 
of $H_{\calL_L}$ as function of $\theta$ is reported in Fig.~\ref{Fig-PattAndSpec}. It has been computed 
for a single pattern with $L=120$ and we have verified that there are no visible variations from one pattern 
to another. The spectrum displays a clear large gap and in fact a second smaller gap is also visible. Upon a 
more careful inspection, both gaps are filled with low density spectrum and so are in fact mobility gaps.

\vspace{0.2cm}

Next, we turn our attention to the Hall conductance \eqref{Eq-HallCond}. In \cite[Ch.~5]{ProdanSpringer2017}, 
a set of very general principles has been formulated for computing correlations of the type seen in \eqref{Eq-HallCond}. 
For disordered crystals, the finite-volume algorithms based on these principles have been shown to converge 
exponentially fast to the thermodynamic limit.   Given the periodic approximates discusses in the previous section, 
the amorphous solid is covered as well by those principles, which we implemented here to 
evaluate Equation \eqref{Eq-HallCond}. For this, we:
\begin{itemize}
\item Computed the Fermi projector $P_F^{(L)}=\chi_{(-\infty,E_F]}(H_{\calL_L})$ using standard routines from functional analysis.
\item Made the commutators compatible with the periodic boundary conditions, by using the optimal 
substitution \cite{ProdanAMRX2013,ProdanSpringer2017}:
\begin{equation}\label{Eq-ApproxDer}
\langle x|[P_F,X_j]|y\rangle \rightarrow \left ((x_j-y_j)-L\left [ \frac{2(x_j-y_j)}{L}\right ]\right ) \, \langle x|P_F^{(L)}|y\rangle,
\end{equation}
where on the right $x$ and $y$ represent the positions of the points inside $[0,L]\times[0,L]$ and 
the square brackets mean the integer part of a real number.
\item Evaluated the relevant matrix elements of the operator inside the trace in \eqref{Eq-HallCond}.
\item Computed the trace per volume using $\Tr_\mathrm{Vol}^{(L)}\{ \cdot \}=\frac{1}{N}\Tr \{\cdot\}$.

\end{itemize}

\begin{figure}
\center
\includegraphics[width=\textwidth]{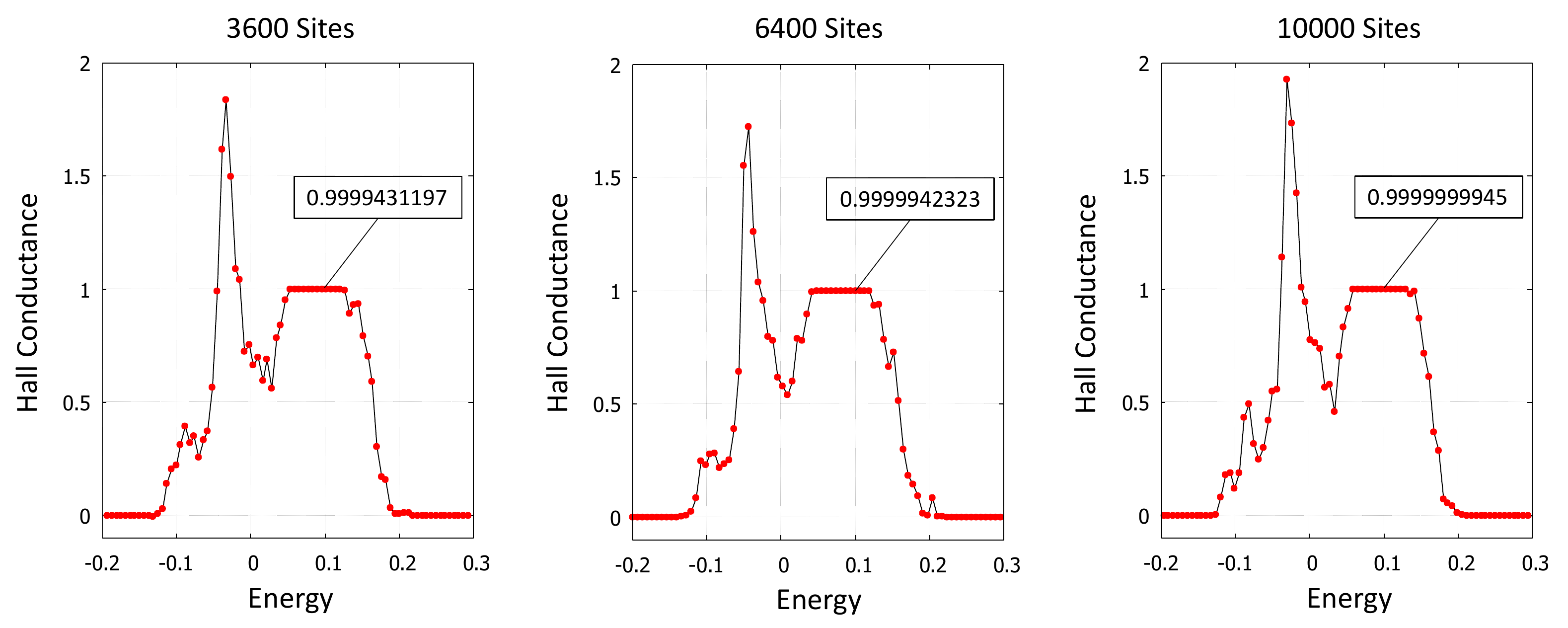}\\
  \caption{\small Plot of the Hall conductance as function of Fermi energy at magnetic field strength 
  $\theta=1.5$ for pattern parameters $d_{\rm min}=0.83$ and (left) $L=60$, (middle) $L=80$, 
  (right) $L=100$. Numerical values are displayed in the boxes.}
 \label{Fig-ChernVsEnergy}
\end{figure}

\vspace{0.2cm}

A map of the Hall conductance as function of Fermi energy and magnetic strength, as computed with the 
above algorithm, is reported in Fig.~\ref{Fig-ChernButterfly} for three increasing system sizes $L=60$, 80 and 100. 
There we can observe a broad band where the Hall conductance takes the quantized value 1 and several 
narrower bands where the Hall conductance takes appreciable values. These bands become sharper as the 
simulation size is increased. The broad band and the first narrow band is consistent with the mobility gaps seen 
in Fig.~\ref{Fig-PattAndSpec}, but notice that the broad band where $\sigma_H=1$ extends much further 
than the region of low spectral density visible with the eye in Fig.~\ref{Fig-PattAndSpec}.

\vspace{0.2cm}

Fig.~\ref{Fig-ChernVsEnergy} shows the sections $\theta \simeq 1.5$ of the intensity maps from 
Fig.~\ref{Fig-ChernButterfly}, together with the numerical values of $\sigma_H$ at $E_F=0.1$. As one can see, 
the quantization is extremely precise even for the smallest system size and for the largest system size the 
quantization holds with eight digits of precision. This is more remarkable given that $\sigma_H$ was computed 
from a single pattern configuration. This leaves very little doubt that, similar to disordered crystals, quantization 
principles are again at work for the amorphous solid; a fact which is confirmed in the following sections. 

\section{Groupoid algebra of a point pattern}
\label{Sec-GrupoidAlg}

\subsection{Delone sets and associated groupoids}

For the reader's convenience, we collect in this section the minimal and quite standard background on point patterns needed for following sections. We take this opportunity to fix our notation. We start by fixing positive numbers $0<r<R<\infty$.
\begin{defn}\label{defn-delone1}
Let $\calL \subset\R^d$ be discrete and infinite and let $B(x;M)$ denote the open ball at $x\in\R^d$ with radius $M>0$.
\begin{enumerate}
  \item $\calL$ is $r$-uniformly discrete if 
    $|B(x;r)\cap \calL| \leq 1$ for all $x\in\R^d$. 
  \item $\calL$ is $R$-relatively dense if 
   $|B(x;R)\cap \calL| \geq 1$ for all $x\in\R^d$.
\end{enumerate}
If $\calL$ is $r$-uniformly discrete and $R$-relatively dense, we call $\calL$ an 
$(r,R)$-Delone set.
\end{defn}

\begin{example} The amorphous patterns constructed in the previous section are Delone sets.$\, \Diamond$
\end{example}

\begin{remark}
Extra structure and properties of Delone sets, {\it e.g.} finite local complexity and repetitive lattices 
can also be considered which potentially give rise to more refined topological properties, 
see~\cite{BBG06} for example. Because our results only require a Delone hypothesis, we 
will not emphasise these extra properties, though we note that patterns of finite local complexity 
do present a strong interest in the study of quasi-crystals and meta-materials. 
$\, \Diamond$
\end{remark}

\begin{prop}[\cite{BHZ00,BBG06}]\label{prop-delonespace}
The set of $(r,R)$-Delone subsets of $\R^d$, $\mathrm{Del}_{(r,R)}$, is a compact and metrizable space, where 
given some $M>0$ and $\epsilon >0$, an $\epsilon$-neighbourhood 
of $\calL$ is given by the set
$$
  U_{M,\epsilon}(\calL) = \big\{ \calL' \in \mathrm{Del}_{(r,R)} \,:\,  d_H\big(\calL\cap B(0;M), \, \calL' \cap B(0;M) \big)  < \epsilon \big\}
$$
with $d_H$ is the Hausdorff distance between sets.
\end{prop}

The space of Delone sets is obviously invariant to the $\R^d$ action $\calL \mapsto \calL + a$ for any $a\in\R^d$. 
These translations act as homeomorphisms in the topology of the space of Delone sets introduced above. 

\begin{defn}
Let $\wt{\calL}$ be an $(r,R)$-Delone subset of $\R^d$. The Hull of $\wt{\calL}$ is the dynamical system 
$(\Omega_{\wt{\calL}}, \R^d, T)$, where $\Omega_{\wt{\calL}}$ is the closure of the orbit of $\wt{\calL}$ 
under the translation action.
\end{defn}

\begin{remarks} 
\begin{enumerate}
  \item Note that $\Omega_{\wt{\calL}}$ is a closed subspace of the space of Delone sets, 
hence it is compact.
  \item To obtain compactness of $\Omega_{\wt{\calL}}$, we actually only require that $\wt{\calL}$ 
  is $r$-uniformly discrete~\cite[Theorem 1.6]{BHZ00}. While many of the results we consider 
  only require this weaker assumption, key results about traces and summability require also 
  an $R$-relatively dense assumption. Hence we generally work with the $(r,R)$-Delone lattices, 
  though we will highlight where this extra condition is required.
  $\, \Diamond$
\end{enumerate}
\end{remarks}

\vspace{0.2cm}

\begin{defn} The transversal of a Delone set $\wt{\calL}$ is given by the set
$$
   \Xi = \{ \calL \in \Omega_{\wt{\calL}}\,:\, 0 \in \calL \},
$$
which is closed and therefore compact.
\end{defn}

\begin{remark} Note that every 
element in $\Xi$ is itself a Delone set. One can think of $\Xi$ as 
the space of (discrete) configurations of an aperiodic lattice.$\, \Diamond$
\end{remark}

\begin{example}
If additional hypotheses are placed on the lattice $\wt{\calL}$, the space 
$\Xi$ can be explicitly characterised. For example, if $\wt{\calL}$ is constructed from 
a Penrose tiling or quasicrystal, $\Xi$ is a Cantor set~\cite{BHZ00}. 
For disordered crystals considered in \cite{PSBbook}, $\Xi$ was homeomorphic with the Hilbert cube.
$\, \Diamond$
\end{example}

A \emph{groupoid} is a small category where all morphisms are invertible. A more user-friendly characterisation 
of a groupoid is a set $\calG$ with an inverse map, 
$\calG \ni \gamma \mapsto \gamma^{-1} \in \calG$, partially defined multiplication, 
$\calG^{(2)} \ni (\gamma_1,\gamma_2)\mapsto \gamma_1\gamma_2 \in \calG$ for $\calG^{(2)} \subset \calG\times \calG$,
and space of units $\calG^{(0)}$. We can define the source and range maps $r,s:\calG\to\calG^{(0)}$ 
as $s(\gamma) = \gamma^{-1}\gamma$ and $r(\gamma)=\gamma\gamma^{-1}$. 
In particular $(\gamma_1,\gamma_2)\in\calG^{(2)}$ if and only if $s(\gamma_1)=r(\gamma_2)$. 
Topological structure can also be added if $\calG$ is a locally compact Hausdorff space, where we require 
the multiplication and inverse maps to be continuous.
A groupoid is called \emph{\'{e}tale} if $r$ is a local homeomorphism.

\begin{prop}[\cite{Kellendonk95}] \label{prop:etale}
Given a Delone set $\wt{\calL}$ and transversal $\Xi$, define the set 
$$
  \calG = \big\{ (\calL, x)\in \Xi \times\R^d \, :\, x\in \calL \big\}.
$$
Then $\calG$ is an \'{e}tale groupoid, where $(\calL,x)^{-1}=(\calL-x,-x)$, $\calG^{(0)}=\Xi$ and 
\begin{align} \label{eq:derivation_properties}
   &s(\calL,x) = \calL - x,   &&r(\calL,x) = \calL,  
   &&(\calL,x)\circ (\calL-x,y) = (\calL, x+y).
\end{align}
\end{prop}

\begin{remark}It is a deep result that when we pass from the continuous 
dynamical system $(\Omega_{\wt{\calL}},\R^d,T)$ to 
the transversal $\Xi$ and groupoid dynamics, the key characteristics of our system are retained~\cite{MRW, SimsWilliams}.$\, \Diamond$
\end{remark}

\vspace{0.2cm}

If the lattice $\wt{\calL}$ is aperiodic, \emph{i.e.} 
there is no $x \neq 0 \in \R^d$ such that $\wt{\calL} - x = \wt{\calL}$, 
then $\calG$ can also be described as the groupoid 
from the \'{e}tale 
equivalence relation on $\Xi\times \Xi$, 
$$
  R_{\Xi} = \big\{(\calL_1,\calL_2)\in \Xi\times\Xi \,:\, 
    \calL_2 =  \calL_1 - a \text{ for some }a\in\R^d \big\}.
$$
Note that the topology on $R_{\Xi}$ is 
different than the subspace topology of $\Xi\times\Xi$. 
For lattices that are not aperiodic, the $C^*$-algebra of the groupoid 
coming from the orbit equivalence relation $R_\Xi$ will not be the correct algebra 
to model a physical system. See~\cite{Kellendonk97} for 
more on these issues.

\subsection{Algebra and representations}

Two groupoid elements $\gamma_1$ and $\gamma_2$ can be composed if 
$s(\gamma_1) = r(\gamma_2)$. Therefore for the case of the transversal groupoid, we can characterize the space of composable elements as
$$
  \calG^{(2)} = \big\{ \big((\calL, x), (\calL - x, y)\big) \big\} \subset \calG \times \calG.
$$ 
One uses this (partial) multiplication to construct a convolution algebra for the groupoid $\calG$. 
Here, the groupoid algebra \cite{Kellendonk95} will be twisted by a cocycle to account for the 
presence of a magnetic field. The interested reader may consult~\cite{Renault80} for a 
comprehensive overview of the general groupoid $C^*$-construction including the twisted case.

\begin{defn}
Let $\calG$ be a locally compact and Hausdorff groupoid. A continuous 
map $\sigma:\calG^{(2)}\to \T$ is a $2$-cocycle if 
\begin{equation}\label{eq-cocyclecond1}
  \sigma(\gamma_1,\gamma_2) \sigma(\gamma_1\gamma_2,\gamma_3) 
    = \sigma(\gamma_1, \gamma_2\gamma_3) \sigma(\gamma_2,\gamma_3)
\end{equation}
for any $(\gamma_1,\gamma_2),(\gamma_2,\gamma_3)\in\calG^{(2)}$, 
and
\begin{equation}\label{eq-cocyclecond2}
  \sigma(\gamma, s(\gamma)) = 1 = \sigma(r(\gamma),\gamma)
\end{equation}
for all $\gamma\in\calG$.
\end{defn}
As the name suggests, groupoid $2$-cocycles give rise to classes in the cohomolgy group 
$H^2(\calG,\T)$, where if $\sigma$ is cohomologous to $\sigma'$, then the corresponding 
(full or reduced) twisted groupoid $C^*$-algebras are isomorphic.  

\vspace{0.2cm}

We encode a magnetic twist on our groupoid via a construction from~\cite{BLM13}, 
which considered twisted crossed products of commutative $C^*$-algebras. 
We construct a magnetic field in $d$ dimensions as a 
$2$-form $B\in\bigwedge^2 \R^d$. Using coordinates 
$B$ can be seen as an anti-symmetric matrix $(B^{j,k})_{j,k=1}^d$ such that 
$$
\partial_j B^{k,l} + \partial_k B^{l,j} + \partial_l B^{j,k} =0.
$$
We then define, for $x,\,y,\,z\in\calL$,
$\Gamma_{\calL}\langle x,y,z \rangle = \int_{\langle x,y,z \rangle} B$
as the magnetic flux through the triangle 
$\langle x,y,z \rangle\subset \R^d\times \R^d$ 
with corners $x,\,y,\,z\in\calL$. In most cases of 
interest, $B$ is a closed $2$-form, $B = \mathrm{d}A$, and so we can 
write the magnetic flux in the more familiar expression
$\Gamma_{\calL}\langle x,y,z \rangle = \int_{\langle x,y,z \rangle} \mathrm{d}A$. 
For this work, we will only consider systems with constant magnetic field 
strength and so the magnetic flux 
can be written using the anti-symmetric matrix $(B^{j,k})$ and 
coordinates of $x,y,z\in\calL$.

With the preliminaries done, we define a magnetic twist via the 2-cocycle 
$\sigma:\calG^{(2)} \to \T$,
$$
  \sigma( (\calL,x),(\calL-x,y)) = \exp\big( -i\Gamma_\calL\langle 0, x, x+y \rangle \big)
$$
as $((\calL,x),(\calL-x,y))\in\calG^{(2)}$ implies $0,\,x,\,x+y\in\calL$.
It is 
straightforward to see that for a $2$-dimensional lattice 
with magnetic field strength $\theta$, 
our twist coincides with Peierls phase factor in \eqref{eq-modelH2}. 
The cocycle condition \eqref{eq-cocyclecond1} on $\sigma$ 
 translates into the condition that, 
for any $x$, $y$ and $z$ such that $x, \ x+y, \ x+y+z\in \calL$,
$$
  \Gamma_{\calL}\langle 0,x,x+y \rangle + \Gamma_{\calL} \langle 0,x+y,x+y+z \rangle 
  = \Gamma_{\calL}\langle 0,x,x+y+z \rangle + \Gamma_{\calL-x}\langle 0, y, y+z \rangle,
$$
which follows from Stokes' Theorem and the observation that 
$$
  \Gamma_{\calL-x}\langle 0, y, y+z \rangle = \Gamma_\calL \langle x , x+y, x+y+z \rangle.
$$
 We also note that our cocycle has the property that 
$\sigma( (\calL,x), (\calL-x,-x)) = 1$ for any $(\calL,x)\in\calG$, which 
will simplify many of our formulas.

\vspace{0.2cm}

Given the groupoid $\calG$ and cocycle $\sigma$, we can construct the twisted convolution 
$\ast$-algebra $C_c(\calG)$ where the elements are functions with compact support over 
$\calG$ and the operations are,
\begin{align*}
  (f_1\ast f_2)(\calL,x) &= \sum_{y\in\calL} f_1(\calL,y) f_2(\calL-y,x-y)
        \,\sigma((\calL,y),(\calL-y,x-y)) \\
    &= \sum_{y\in\calL} e^{-i\Gamma_\calL\langle 0, y, x \rangle } 
         f_1(\calL,y) f_2(\calL-y,x-y) \\
   f^*(\omega,x) &= \ol{f(\calL-x,-x) \sigma((\calL,x),(\calL-x,-x)) } 
     = \ol{f(\calL-x,-x)}
\end{align*}
The cocycle condition \eqref{eq-cocyclecond1} on $\sigma$ ensures that $C_c(\calG)$ is associative 
and it is a simple check that $(f_1\ast f_2)^\ast = f_2^*\ast f_1^*$. The algebra $C_c(\calG)$ is unital with the unit $1(\calL,x)= \delta_{x,0}$. Furthermore, it accepts a family of canonical representations, $\{\pi_\calL\}_{\calL\in\Xi}$, indexed by 
$\calL\in\Xi$ and defined by the maps $\pi_\calL:C_c(\calG) \to \calB[\ell^2(\calL)]$,
\begin{equation} \label{eq:Hspace_repn}
  \big(\pi_\calL(f)\psi)(x) = \sum_{y\in\calL}e^{-i\Gamma_{\calL-x}\langle 0,y-x,-x \rangle} 
     f(\calL-x,y-x) \psi(y).
\end{equation}
One can check that $\pi_\calL(f_1\ast f_2) = \pi_\calL(f_1)\pi_\calL(f_2)$ and 
$\pi_\calL(f^*)=\pi_\calL(f)^*$ so $\pi_\calL$ is indeed a $\ast$-representation. Also, $\pi_\calL(1)=1_{\calB[\ell^2(\calL)]}$.

\begin{remark} With the substitution $q=y-x$, \eqref{eq:Hspace_repn} becomes:
$$
  \big(\pi_\calL(f)\psi)(x) = \sum_{q\in\calL-x}e^{-i\Gamma_{\calL-x}\langle 0,q,-x \rangle} 
     f(\calL-x,q) \psi(x+q),
$$
which, apart from the Peierls factor, is identical to \eqref{Eq-StrHamiltonian} if the 
coefficients are properly identified. In other words, the canonical representations of $C_c(\calG)$ generate 
all the finite range Hamiltonians associated to $\calL$.$\, \Diamond$  
\end{remark}

\vspace{0.2cm}

The next result gives a covariance for representations that come from lattices in the same orbit, 
which needs to be verified with care when the cocycle is present. 

\begin{prop} \label{prop:repn_covariance}
Suppose that $\calL,\calL'\in\Xi$ are such that $\calL'=\calL-a$ for some $a\in \R^d$. 
Then there is a unitary operator $T_a:\ell^2(\calL)\to\ell^2(\calL-a)$ such that 
$T_a \pi_\calL(f) T_a^* = \pi_{\calL-a}(f)$ for all $f\in C_c(\calG)$.
\end{prop}
\begin{proof}
Because our proof relies on the cocycle condition, we will work 
with the cocycle $\sigma$ directly. First we define a unitary maps 
$T_a:\ell^2(\calL)\to\ell^2(\calL-a)$ which stand for the magnetic shifts, where
$$
  (T_a\psi)(x) = \sigma((\calL-x-a,-x-a),(\calL,a)) \psi(x+a) 
    = e^{-i\Gamma_{\calL-x-a} \langle 0, -x-a,x \rangle } \psi(x+a).
$$
One then checks that the inverse $T_a^*:\ell^2(\calL-a)\to \ell^2(\calL)$ 
can be written in the form 
$$
  (T_a^* \phi)(y) = \sigma((\calL-y,-y),(\calL,a) )^{-1} \phi(y-a) 
   = e^{i\Gamma_{\calL-y}\langle 0, -y,a-y \rangle } \phi(y-a).
$$
We now verify the compatibility of our representation 
with this unitary map.
\begin{align*}
 \big( T_a \pi_\calL(f) T_a^* \psi \big)(x) &= \sigma((\calL-x-a,-x-a),(\calL,a)) 
   \big( \pi_\calL(f) T_a^* \psi \big)(x+a) \\
   &\hspace{-1cm}=  \sigma((\calL-x-a,-x-a),(\calL,a))  \sum_{y\in\calL} 
       \sigma((\calL-x-a,y-x-a),(\calL-y,-y)) \\
     &\qquad \times f(\calL-x-a,y-x-a) (T_a^* \psi)(y) \\
   &\hspace{-1cm}= \sigma((\calL-x-a,-x-a),(\calL,a))  \sum_{y\in\calL} 
       \sigma((\calL-x-a,y-x-a),(\calL-y,-y)) \\
     &\qquad \times f(\calL-x-a,y-x-a) 
        \sigma((\calL-y,-y),(\calL,a) )^{-1} \psi(y-a) \\
   &\hspace{-1cm}=  \sum_{u\in\calL-a} f(\calL-x-a,u-x)\psi(u) 
          \sigma((\calL-x-a,u-x),(\calL-a-u,-u-a)) \\
     &\qquad \times 
       \sigma((\calL-x-a,-x-a),(\calL,a)) \sigma((\calL-a-u,-u-a),(\calL,a))^{-1}  \\
   &\hspace{-1cm}= \sum_{u\in\calL-a}\sigma((\calL-a-x,u-x),(\calL-a-u,-u)) f(\calL-x-a,u-x)\psi(u) \\
   &\hspace{-1cm}= \big(\pi_{\calL-a}(f)\psi\big)(x),
\end{align*}
where in the second to last line we have used the cocycle identity
\begin{equation*}
  \sigma(\gamma_1,\gamma_2\gamma_3) = \sigma(\gamma_1,\gamma_2) 
     \sigma(\gamma_1\gamma_2,\gamma_3)\sigma(\gamma_2,\gamma_3)^{-1}, \qquad 
     (\gamma_1,\gamma_2), (\gamma_2,\gamma_3)\in\calG^{(2)}. \qedhere
\end{equation*}
\end{proof}

\begin{defn}
The twisted reduced groupoid $C^*$-algebra $C^*_r(\calG,\sigma)$ is given 
by the $C^*$-completion of $C_c(\calG)$ under the norm
$$ 
  \|f\| = \sup_{\calL\in\Xi} \|\pi_\calL(f)\|.
$$
\end{defn}

\begin{remark}
The family of representations $\{\pi_\calL\}_{\calL\in\Xi}$ of $C_c(\calG)$ extends to a family of 
representations of the $C^*$-closure $C^*_r(\calG,\sigma)$. In particular, the representations  
of the $C^*$-closure represent Hamiltonians associated to $\calL\in \Xi$ without 
the finite range assumption. Hence Hamiltonians such as Equations \eqref{Eq-ModelHam1} 
and \eqref{eq-modelH2} are represented in the $C^*$-closure. 
$\, \Diamond$
\end{remark}

\begin{remark}
We can easily extend our framework to the case of representations and Hamiltonians on 
$\C^N\otimes \ell^2(\calL)$ by working with the algbera $M_N(C^*_r(\calG,\sigma))$ 
of $N\times N$ matrices with entries in the $C^*$-algebra. In order to keep our presentation 
as clean as possible, we write the case of $N=1$ but note that our index theory 
results also apply to systems with $N$ degrees of freedom described in 
Section \ref{Sec:Patterned_resonators} by this matrix extension.
$\, \Diamond$
\end{remark}

\subsection{Differential calculus}

In this section we introduce a set of canonical derivations and invariant trace for $C^*_r(\calG,\sigma)$. 
Along the way, we introduce the smooth and Sobolev algebras which will play central roles. 

\subsubsection{Derivations and the smooth subalgebra} 
We would like to encode a differential structure on the dense subalgebra 
$C_c(\calG) \subset C^*_r(\calG,\sigma)$. To do this we first note that 
the algebra $C_c(\calG)$ has a family of  $d$ commuting 
one-parameter group of automorphisms 
$\{u_t^{(j)}\}_{j=1}^d$, where
$$
   (u_t^{(j)} f)(\calL,x) = e^{it x_j} f(\calL,x), \qquad t\in\R
$$ 
and $x_j$ the $j$-th component of $x\in\calL$.
The generators of these automorphisms are the derivations
$\{\partial_j\}_{j=1}^d$ on $C_c(\calG)$, 
where $(\partial_j f)(\calL,x) = x_j f(\calL,x)$ (pointwise multiplication). Similar 
groupoid dynamics appear in~\cite{MeslandGpoid}.
The representations $\{\pi_\calL\}_{\calL\in\Xi}$ relate the derivations $\partial_j$ 
on the algebra to the unbounded position operator $X_j:\Dom(X_j)\subset \ell^2(\calL)\to\ell^2(\calL)$ on 
the Hilbert space. Namely, basic computations give that 
\begin{align}  \label{eq:conv_derivation}
  &\partial_j(f_1\ast f_2) = f_1 \ast \partial_j f_2 + \partial_j f_1 \ast f_2, 
  && \pi_\calL(\partial_j f) = [X_j,\pi_\calL(f)], 
\end{align}
for any $j\in \{1,\ldots,d\}$ and $\calL\in\Xi$. We note that the operators $\{X_j\}_{j=1}^d$ 
also depend on the lattice $\calL$. We will slightly abuse notation and refer to the operator 
$X_j$ as the $j$-th position operator on any lattice $\calL\in\Xi$, where the particular 
space in which $X_j$ acts will be clear from the context.

\begin{remark} Note that it is precisely the commutators of Equation \eqref{eq:conv_derivation} that enter the 
expression of the Hall conductance \eqref{Eq-HallCond}. The link between these 
commutators and the derivations on the algebra was paramount for finding the optimal 
substitution in \eqref{Eq-ApproxDer} (see \cite[Sec.~4.5]{ProdanSpringer2017}).
$\, \Diamond$
\end{remark}

\vspace{0.2cm}

We note that $\partial_j(C_c(\calG)) \subset C_c(\calG)$ and so our subalgebra $C_c(\calG)$ is 
`smooth' under the derivations $\{\partial_j\}_{j=1}^d$. 
However, from the perspective of index theory, we need to make sure that we do not 
lose any information when working with a subalgebra, where for example the Fermi projection does not belong to 
$C_c(\calG)$ even if $E_F$ is located in a spectral gap of a finite range Hamiltonian. 
This problem is solved by completing $C_c(\calG)$ in a topology stronger than the $C^*$-norm so 
that elements in the completion remain sufficiently smooth, yet the completion is large enough so 
that all $K$-theoretic results extend to the $C^*$-algebra $C^*_r(\calG,\sigma)$.

\begin{defn}
The smooth algebra $\calA$ is defined as the completion of $C_c(\calG)$ under the topology induced by the norms
$$
   \| f \|_\alpha = \left\| \partial^\alpha f \right\|,  
    \qquad \partial^\alpha = \partial_1^{\alpha_1}\cdots \partial_d^{\alpha_d},  \quad \alpha \in \N^d.
$$
\end{defn}

\begin{prop}[\cite{RennieSmooth, PSBbook}]
The algebra $\calA$ is Fr\'{e}chet and stable under the holomorphic functional calculus. 
In particular $K_\ast(\calA) \cong K_\ast (C^*_r(\calG,\sigma))$, with $\ast=0,1$.
\end{prop}

\begin{remark}\label{re-smoothproj} One can improve the above statement by observing that in fact $\calA$ is 
invariant to the smooth functional calculus \cite[Prop.~3.25]{ProdanSpringer2017}. 
Then one can automatically see that the spectral projections of any self-adjoint element 
from $\calA$ are also elements of $\calA$, provided the edges of the spectral intervals 
are located inside spectral gaps.$\, \Diamond$
\end{remark}

\subsubsection{Traces and Sobolev spaces}

Our next task is to construct a trace on our groupoid algebra. The following 
result will be useful.

\begin{prop}[\cite{BBG06, SavThesis}]
There is a one-to-one correspondence between measures on $\Omega_{\wt{\calL}}$, the continuous hull of $\wt{\calL}$, 
invariant under the $\R^d$-action and measures on the transversal $\Xi$ invariant under the groupoid action. 
\end{prop}

Every topological dynamical system admits faithful, normalized, invariant and ergodic measures. From  now on, 
we fix such a measure on $\Omega_{\wt{\calL}}$, which in turn gives a measure $\bP$ on $\Xi$. Using $\bP$,  
we can define the dual faithful trace
$$
   \calT(f) = \int_{\Xi} f(\calL,0)\,\mathrm{d}\bP(\calL), 
   \qquad f\in \calA.
$$
The following statement gives physical meaning to the trace we defined above.
\begin{prop} \label{prop:ergodic_trace_is_vol_trace}
For all $f \in \calA$ and $\bP$-almost all $\calL\in\Xi$,
$$
   \calT(f) = \Tr_\mathrm{Vol}( \pi_\calL(f)),
$$
where $\Tr_\mathrm{Vol}$ is the trace per unit volume in $\ell^2(\calL)$.
\end{prop}
\begin{proof}
Recall the representation
$$
  (\pi_\calL(f)\psi)(x) = \sum_{y\in\calL}
   e^{-i\Gamma_{\calL-x}\langle 0,y-x,-x \rangle} f(\calL-x,y-x) \psi(y). 
$$
On a finite sublattice $\Lambda \subset\calL$, $\pi_\calL(f)$ is trace-class 
and we can compute 
$$
  \Tr_{\Lambda}(\pi_\omega(f)) = \sum_{x\in\Lambda} 
    e^{-i\Gamma_{\calL-x}\langle 0,0,-x \rangle} f(\calL-x,x-x) 
  = \sum_{x\in\Lambda} f(\calL-x,0).
$$
Then, by the $R$-relative density of $\calL$ and Birkhoff's ergodic theorem \cite{BirKhoff1931},
$$
  \Tr_\mathrm{Vol}( \pi_\calL(f)) = \lim_{\Lambda \to \calL} 
  \frac{1}{|\Lambda|}\Tr_\Lambda(\pi_\calL(f)) 
    = \lim_{\Lambda \to \calL}\frac{1}{|\Lambda|} \sum_{x\in\Lambda} f(\calL-x,0) 
   = \int_\Xi f(\calL,0)\,\mathrm{d}\bP(\calL),
$$
which gives the result.
\end{proof}

\begin{remark} The Hall conductance \eqref{Eq-HallCond} can be expressed now without 
involving any Hilbert spaces but only the algebra of physical observables and its differential 
calculus. Indeed, let $h\in \calA$ be the element which generates the Hamiltonians 
associated to pattern $\calL$ and let $p_F$ be its Fermi projection (assuming 
$E_F$ in a spectral gap). Then (up to a physical constant)
\begin{equation}\label{eq-hallcond2}
\sigma_H= \calT(p_F [\partial_1 p_F,\partial_2 p_F]),
\end{equation} 
which is a direct translation of \eqref{Eq-HallCond}.$\, \Diamond$
\end{remark}

\vspace{0.2cm}

The trace $\calT$ gives us the GNS Hilbert space $L^2(\calG,\calT)$, which 
is the completion of $C_c(\calG)$ under the inner-product 
$\langle f_1,f_2 \rangle = \calT(f_1^*\ast f_2)$. The space $L^2(\calG,\calT)$ has the 
canonical representation of $C^*_r(\calG,\sigma)$ given by the extension of 
left-multiplication. Namely, we take the $C^*$-completion of the following action
$$
  \pi_{GNS}(f_1) f_2 = f_1 \ast f_2, \qquad f_1\in C_c(\calG), \quad f_2\in C_c(\calG) \subset L^2(\calG,\calT).
$$
Also of importance is the von Neumann algebra $L^\infty(\calG,\calT)$, which is the weak 
closure of the GNS representation of $C^*_r(\calG,\sigma)$ in $L^2(\calG,\calT)$. 
The von Neumann algebra $L^\infty(\calG,\calT)$ comes with the norm
\begin{equation}\label{eq-vneumannnorm}
\|f\|_{L^\infty} = \bP -\mathrm{ess.}\, \sup_{\calL\in\Xi} \|\pi_\calL(f)\|.
\end{equation}

\vspace{0.2cm}

When the Fermi level is not located inside a spectral gap, the Fermi projection is not 
even an element of $C_r^\ast(\calG,\sigma)$. In contrast, all spectral projections are elements 
of $L^\infty(\calG,\calT)$ since von Neumann algebras are invariant to the Borel functional calculus. 
When the Fermi levels are located in mobility gaps, then certain correlations become finite and 
the Fermi projections belong to a strict subalgebra of $L^\infty(\calG,\calT)$. 
We call this subalgebra the Sobolev algebra, which is defined using the 
derivations $\{\partial_j\}_{j=1}^d$ and 
the theory of non-commutative $L^p$-spaces 
associated to the von Neumann algebra $L^\infty(\calG,\calT)$ and trace $\calT$. 
The $L^p$-spaces are the Banach spaces given by the completion of $C_c(\calG)$ under the norms
$$
   \| f \|_p = \calT\big( |f|^p \big)^{1/p}, \qquad |f| = \sqrt{f^\ast \ast f}, \quad p\in [1,\infty).
$$
We can control the $L^p$-norm of products by using the H\"{o}lder's inequality, see~\cite{FK}.
\begin{equation}\label{eq-holder}
  \|f_1 \cdots f_k \|_p \leq \|f_1\|_{p_1} \cdots \|f_k\|_{p_k}, \qquad 
   \frac{1}{p_1} + \cdots + \frac{1}{p_k} = \frac{1}{p}.
\end{equation}

\begin{defn}
The Sobolev spaces $\calW_{r,p}$ are defined as the Banach spaces
obtained as the completion of $C_c(\calG)$ in the norms 
$$
  \| f\|_{r,p} = \sum_{|\alpha| \leq r} \calT\Big( |\partial^\alpha f|^p \Big)^{1/p}, \qquad r\in \N, \,\, p\in [1,\infty),
$$
where we use multi-index notation, $\alpha\in \N^d$, 
$\partial^\alpha = \partial_1^{\alpha_1}\partial_2^{\alpha_2}\cdots \partial_d^{\alpha_d}$ 
and $|\alpha| = \alpha_1 +\cdots +\alpha_d$.
\end{defn}

The Sobolev spaces are not closed under multiplication but taking their intersection gives an algebra structure.

\begin{defn}
The Sobolev algebra $\calA_\mathrm{Sob}$ is defined as the intersection of 
$L^\infty(\calG,\calT)$ with the Fr\'{e}chet algebra that comes from the 
completion of $C_c(\calG)$ in the topology defined by the norms 
$\|\cdot\|_{r,p}$ for $r\in\N$, $p\in \N_+$.
\end{defn}

The algebras $\calA_\mathrm{Sob}$ and $L^\infty(\calG,\calT)$ are naturally embedded 
within $\calB[L^2(\calG,\calT)]$, though we can also consider the representations of 
$L^\infty(\calG,\calT)$ on $\ell^2(\calL)$. Indeed, by \eqref{eq-vneumannnorm}, 
$f \in L^\infty(\calG,\calT)$ if and only if $\pi_\calL(f)$ defined in \eqref{eq:Hspace_repn}  
lands in $\calB[\ell^2(\calL)]$ for all $\calL \in \Xi$ less a set of zero $\bP$-measure. 
Note that this zero measure set changes from one element to another but, nevertheless, 
if $f,g \in L^\infty(\calG,\calT)$, then there is a subset of measure one in $\Xi$ where 
$\pi_\calL(f)$, $\pi_\calL(g)$ and $\pi_\calL(f\ast g)$ all land in $\calB[\ell^2(\calL)]$ 
and, on that subset, $\pi_\calL(f \ast g) = \pi_\calL(f)\, \pi_\calL(g)$. 
Furthermore, if this applies to $\calL$, then it applies to the entire orbit of $\calL$ in $\Xi$. 
In this sense, and only in this sense, we can consider the family $\{\wt{\pi}_\calL\}_{\calL\in\Xi}$ of representations of
the von Neumann algebra. We also note that 
by the canonical extension of $\calT$ to $L^\infty(\calG,\calT)$, an analogous version 
of Proposition \ref{prop:ergodic_trace_is_vol_trace} holds for the representations 
$\{\wt{\pi}_\calL\}_{\calL\in\Xi}$ as the work of Birkhoff \cite{BirKhoff1931} applies to measurable functions too.

\section{The Dirac spectral triples}

\subsection{The smooth version}

We use the differential structure on the smooth algebra $\calA$ to construct a Dirac-like 
operator and spectral triple, see the appendix for basic definitions and properties.
To put everything together, we use the (trivial) spin$^c$ structure on 
$\R^d$. Namely, for $\C^\nu$ with $\nu = 2^{\lfloor \frac{d}{2} \rfloor}$ 
there exist self-adjoint matrices $\{\Gamma^j\}_{j=1}^d\subset M_\nu(\C)$ such that 
$\Gamma^j\Gamma^k + \Gamma^k\Gamma^j = 2\delta_{j,k}\cdot 1_\nu$. The matrices 
$\{\Gamma^j\}_{j=1}^d$ can be constructed via tensor products of the $2\times 2$ 
Pauli matrices (see~\cite[Appendix A]{GSB16} for example). When 
$d$ is even, $\C^\nu$ is a graded vector space with grading operator  
$\Gamma_0 = (-i)^{d/2}\Gamma^1\cdots\Gamma^d$. 
Using this irreducible 
Clifford representation, we can construct the unbounded operator 
$X = \sum_{j=1}^d X_j\hat\otimes \Gamma^j$ on $\ell^2(\calL)\hat\otimes \C^\nu$. 
We can also diagonally extend the representations of our various algebras to representations on 
$\ell^2(\calL) \hat\otimes \C^\nu$.

\begin{prop} \label{prop:smooth_spectrip}
For any $\calL\in\Xi$, the triple
$$
  \lambda_{d}^\calL = 
  \bigg( \calA, \, {}_{\pi_\calL}\ell^2(\calL) \hat\otimes \C^\nu, \, 
    X= \sum_{j=1}^d X_j \hat\otimes \Gamma^j \bigg)
$$
is a $QC^\infty$ and $d$-summable spectral triple.
\end{prop}
\begin{proof} 
We first verify the defining properties stated in Definition \ref{defn-spectriple}. 
The strong regularity of elements in $f\in\calA$ in the spacial coordinate 
ensures that when we take the representation $\pi_\calL(f)\Dom(X) \subset \Dom(X)$.
Recall also Equation \eqref{eq:conv_derivation} which gives
\begin{equation}\label{eq-commutator1}
  [X, \pi_\calL(f)] = \sum_{j=1}^d [X_j,\pi_\calL(f)]\hat\otimes \Gamma^j = \sum_{j=1}^d \pi_\calL(\partial_j f)\hat\otimes \Gamma^j.
\end{equation}
Since $\calA\subset C_r^\ast(\calG,\sigma)$ is invariant under derivations,  the result is indeed a bounded operator for any $f\in\calA$. 

\vspace{0.2cm}

Next we note that 
$(1+X^2) = (1+|X|^2) \hat\otimes 1_{\C^{\nu}}$.
In the canonical basis $\{e_y\}_{y\in\calL}$ of $\ell^2(\calL)$, we have that 
$(1+ X^2)e_y\otimes \xi = (1+|y|^2)e_y \otimes \xi$ for all $\xi \in \C^\nu$. Therefore we can decompose
\begin{equation}\label{eq-rezdec}
  (1+ X^2)^{-1/2} = \sum_{y\in\calL} (1+|y|^2)^{-1/2} |e_y\rangle \langle e_y | \otimes 1_\nu.
\end{equation}
Because $\calL$ is $R$-relatively dense, 
we can write $(1+X^2)^{-1/2}$ as a norm-convergent sum of finite-rank operators. Therefore it is compact. 

\vspace{0.2cm}

The spectral triple is $QC^\infty$ (see Definition~\ref{defn-qccondition}), since $\calA$ is invariant to derivations of any order. 
Lastly, we verify summability (see Definition~\ref{defn-summability}). For this, we compute 
\begin{align*}
   \Tr((1+X^2)^{-s/2}) =  \sum_{y\in\calL} (1+|y|^2)^{-s/2} (\Tr_{\ell^2(\calL)} \otimes \Tr_{\C^\nu})( |e_y\rangle \langle e_y | \otimes 1_\nu ) 
     = \nu \sum_{y\in\calL} (1+|y|^2)^{-s/2},
\end{align*}
which is finite for $s>d$ and any Delone set $\calL\subset \R^d$.\footnote{We note that the compactness 
and summability of  $(1+X^2)^{-1/2}$ 
may fail if $\calL$ is only $r$-uniformly discrete and not $(r,R)$-Delone.}
\end{proof}

\begin{prop} \label{prop:orbit_equiv_spec_trip}
Suppose that $\calL,\calL'\in\Xi$ are such that $\calL'=\calL-a$, 
then the spectral triples $\lambda_d^{\calL}$ and $\lambda_d^{\calL'}$ 
define the same class in the $K$-homology of $C_r^*(\calG,\sigma)$.
\end{prop}
\begin{proof}
Recall Proposition \ref{prop:repn_covariance}, which defined a unitary operator 
$T_a: \ell^2(\calL) \to \ell^2(\calL-a)$ such that 
$T_a \pi_\calL(f) T_a^* = \pi_{\calL-a}$. Applying this unitary map to our 
spectral triple, we induce a shift in the unbounded operator 
$T_a X_j T_a^* = X_j +a_j$. Therefore, the isomorphism $T_a$ gives the 
unitarily equivalent spectral triple
$$
 \bigg( \calA, \, {}_{\pi_{\calL-a}}\ell^2(\calL-a)\hat\otimes \C^\nu, \, 
   \sum_{j=1}^d (X_j + a_j)\hat\otimes \Gamma^j \bigg).
$$ 
We can then take an operator homotopy 
$X_t = \sum_{j=1}^d (X_j + (1-t)a_j)\hat\otimes \Gamma^j$ for $t\in[0,1]$. 
This homotopy then directly connects us to $\lambda^{\calL-a}_d$. 
Taking the bounded transform, the $K$-homology classes of equivalent 
spectral triples will  coincide.
\end{proof}

\begin{cor} 
$\bP$-almost surely, all spectral triples $\lambda^\calL_d$ define the same $K$-homology class. 
As such, the index pairings (Definition \ref{defn-indexpairing}) return $\bP$-almost surely the same integer number.
\end{cor}
\begin{proof}
Our working hypothesis that $\bP$ is ergodic implies that all elements $\calL\in\Xi$ are 
almost surely orbit equivalent. The result then follows from 
Proposition \ref{prop:orbit_equiv_spec_trip}.
\end{proof}

\subsection{The Sobolev version}

We can also construct a spectral triple from the much larger Sobolev algebra. This spectral triple 
 will retain finite summability 
and enough regularity so that we can extend the index pairing.

\begin{prop} \label{prop:sobolev_spectrip}
The family 
$$
 \wt \lambda_{d}^\calL = 
       \bigg( \calA_\mathrm{Sob}, \, {}_{\wt{\pi}_\calL}\ell^2(\calL) \hat\otimes \C^\nu, \, \sum_{j=1}^d X_j\hat\otimes\Gamma^j \bigg),
$$
indexed by $\calL \in (\Xi,\bP)$, is a $\bP$-almost sure family of spectral triples 
(see Definition~\ref{defn-sobspectriple}), which is $\bP$-almost surely $QC^m$ and 
$d$-summable for $m=\mathrm{max}\{2,d-2\}$. 
\end{prop}
\begin{proof}
As in Proposition \ref{prop:smooth_spectrip}, we find
$$
  [X, \wt \pi_\calL(f)] = \sum_{j=1}^d [X_j,\wt\pi_\calL(f)]\hat\otimes \Gamma^j = \sum_{j=1}^d \wt \pi_\calL(\partial_j f)\hat\otimes \Gamma^j.
$$
Since the Sobolev algebra is invariant under derivations, 
$\sum_{j=1}^d \partial_j f\hat\otimes \Gamma^j \in L^\infty(\calG,\calT) \otimes \C^\nu$ and 
from \eqref{eq-vneumannnorm} we see that the commutator is $\bP$-almost surely bounded. 
Note, however, that the zero-measure set where the commutator may be unbounded depends 
on the element $f\in \calA_{\rm Sob}$.\footnote{Indeed, this observation is what motivated our 
definition of a $\bP$-almost sure family of spectral triples.}

\vspace{0.2cm}

Because our Hilbert space and operator $X$ are the same as the smooth case, the decomposition \eqref{eq-rezdec} 
still holds. Therefore compactness of the resolvent and finite summability then carries over. 
Lastly, we recall that the family of spectral triples is $\bP$-almost surely $QC^m$ if, $\bP$-almost surely,
\begin{align*}
     &\pi_\calL(f), \pi_\calL(\partial_j f) \in \bigcap_{k\leq m} \Dom(\delta^k),  
     &&\delta(T) = [|X|,T], 
     &&(|X|\psi)(x) = |x|\psi(x)
\end{align*}
for any $j\in\{1,\ldots,d\}$ and $f \in \calA_{\rm Sob}$. Simple computations give that for $a\in\calA_\mathrm{Sob}$ considered 
as a measurable function of $\calG$,
$$
   \big(\delta^m(\pi_\calL(a))\psi\big)(x) = \sum_{y\in\calL} e^{-i\Gamma_{\calL-x}\langle 0, y-x, -x \rangle} 
   \big(|x| - |y|\big)^m a(\calL-x,y-x) \psi(y).      
$$
Using the bound $|x|-|y|\leq |x-y|$, the result will follow if we can show that 
$|\partial |^m a \in \calA_\mathrm{Sob}$ and $|\partial |^m (\partial_j a) \in \calA_\mathrm{Sob}$ 
for $a\in \calA_\mathrm{Sob}$ and $(|\partial |a)(\calL,x) = |x| a(\calL,x)$ a partial derivation. 
We then note that 
$$
  \| |\partial |^m(a)\|_{r,p} \leq C_m \|a\|_{r+m,p}, \qquad
  \| |\partial |^m(\partial_ja) \|_{r,p} \leq C_m \|a\|_{r+m+1,p}
$$
as required.
\end{proof}

For regular spectral triples such as $\lambda_{d}^\calL$ from Proposition \ref{prop:smooth_spectrip}, 
there is a well-defined $\Z$-valued pairing with $K$-theory elements. We now construct 
an analogous pairing for the almost sure family of 
spectral triples $\{\wt \lambda_{d}^\calL \big \}_{\calL \in \Xi}$ with $K_*^{\rm Alg}(\calA_{\rm Sob})$. 
Let us point out that, for separable and stabilized $C^\ast$, Banach and classes of Fr\'{e}chet algebras, the algebraic and 
topological K-theories are isomorphic \cite{Cortinas2008}, but $\calA_{\rm Sob}$ is 
not separable and most likely such relation can not be established. 
For example, it is known that the dimension function associated to the trace 
$\calT$ changes under continuous deformations of projections inside $\calA_{\rm Sob}$. 

\begin{prop}\label{prop-sobindexpairing} 
The integer pairings in Definition~\ref{defn-indexpairing} can be extended to integer 
pairings between the family $\big \{\wt \lambda_{d}^\calL \big \}_{\calL \in \Xi}$ and the 
appropriate $K_\ast^{\rm Alg}$-groups of $\calA_{\rm Sob}$.
\end{prop}

\begin{proof}  
For $d$ even and $p\in M_N(\calA_\mathrm{Sob})$ a projection, we show that the $\bP$-almost sure index
\begin{equation}\label{eq-sobpairing1}
  \Big \langle [p], \big \{\wt \lambda_{d}^\calL \big \}_{\calL \in \Xi} \Big \rangle := 
    \Index\big( \wt \pi_\calL(p)( F_X \otimes 1_N)_+ \wt \pi_\calL(p)\big), \quad F_X=\frac{X}{\sqrt{1+X^2}},
\end{equation}
is a well defined map $K_0^{\rm Alg}(\calA_\mathrm{Sob})\to \Z$, where $( F_X \otimes 1_N)_+$ indicates the bottom-left corner 
 of the operator $F_X\otimes 1_N$ when we 
decompose in the grading $1_N \otimes \Gamma_0$. For $d$ odd and $u\in M_N(\calA_\mathrm{Sob})$ unitary, we consider 
\begin{equation}\label{eq-sobpairing2}
  \Big \langle [u], \big \{\wt \lambda_{d}^\calL \big \}_{\calL \in \Xi} \Big \rangle := 
    \Index\big( \Pi_N \wt \pi_\calL(u) \Pi_N - (1-\Pi_N) \big). \quad \Pi_N = \tfrac{1}{2}(1+F_X)\otimes 1_N,
\end{equation}
as a $\bP$-almost sure map $K_1^{\rm Alg}(\calA_\mathrm{Sob})\to \Z$.

\vspace{0.2cm}

First we observe that $\bP$-almost surely, the operators inside Index in \eqref{eq-sobpairing1} and \eqref{eq-sobpairing2} are Fredholm. 
This is a pure functional analytic which follows from the fact that $F_X$ is Fredholm 
and that $[\wt\pi_\calL(p),F_X]$ or $[\wt\pi_\calL(u),F_X]$ are $\bP$-almost surely 
bounded and compact. Then indeed, the relevant operators are invertible up to compacts. 
Furthermore, if the Fredholm index is well defined for $\calL \in \Xi$, then it is well 
defined and constant for the entire orbit of $\calL$. 
Indeed, for $\calL,\calL'\in\Xi$ with $\calL'=\calL-a$,
$$
  \Index(\wt{\pi}_\calL(p) (F_X\otimes 1_N)_+ \wt{\pi}_\calL(p)) 
   = \Index( \wt{\pi}_{\calL-a}(p) (F_{X+a}\otimes 1_N)_+ \wt{\pi}_{\calL-a}(p)),
$$
where $F_{X+a}$ is the bounded 
transform of $\sum_j (X_j+a_j)\otimes \Gamma^j$.
The bounded perturbation of $X$ by $a$ implies that the perturbation 
$F_{X+a}-F_X$ is compact. Therefore
\begin{align*}
  \Index(\wt{\pi}_\calL(p) (F_X\otimes 1_N)_+ \wt{\pi}_\calL(p)) &= \Index( \wt{\pi}_{\calL-a}(p) (F_{X}\otimes 1_N)_+ \wt{\pi}_{\calL-a}(p) + K ) \\
    &= \Index( \wt{\pi}_{\calL-a}(p) (F_{X}\otimes 1_N)_+ \wt{\pi}_{\calL-a}(p))
\end{align*}
for $K$ compact. The odd index follows the same argument.
Since the measure is ergodic with respect to the translations, the right hand sides 
of \eqref{eq-sobpairing1} and \eqref{eq-sobpairing2} are $\bP$-almost surely well defined and take value in $\Z$.

\vspace{0.2cm}

Now suppose that $[p]=[p']$ in $K_0^{\rm Alg}(\calA_\mathrm{Sob})$ and so there exists an invertible 
element $w \in M_\infty(\calA_{\rm Sob})$ such that $p' = w^{-1}p w$. Then the invariance of the index 
under conjugation by an invertible and the fact that 
$\wt \pi_\calL(w^{-1})F_X\pi_\calL(w) -F_X$ is $\bP$-almost surely 
compact ensures that the index is constant over the $K_0^{\rm Alg}$-classes. 
The additive property of the Fredholm index then ensures that our map is a 
well-defined group homomorphism $K_0^{\rm Alg}(\calA_\mathrm{Sob}) \to \Z$. 
The odd case follow from similar arguments and we obtain a group 
homomorphism $K_1^{\rm Alg}(\calA_\mathrm{Sob})\to\Z$.
\end{proof}

\section{The local index formulas}
\label{Sec-IndexFormulas}

In this section we derive local formulas for the index pairings defined in the previous section. 
The starting points for our calculations are the general local index theorems in 
non-commutative geometry~\cite{CoM,CPRS2,CPRS3}, which are re-stated 
Section \ref{subsec_localindexappendix} for the reader's convenience. 
An important step still remains to be completed if we 
want to connect the general index formulas with the physical response coefficients of a system, 
such as the Hall conductance \eqref{eq-hallcond2}. For the spectral triples 
we consider and after some algebraic manipulation, 
the index formulas reduce to the computation of a residue trace. 
 We apply this formula first to the smooth case, 
which fits into the standard setting of the local index theorems, and then show that the results can be 
pushed into the regime of a mobility gap.

\subsection{The smooth case}

In Proposition~\ref{prop:smooth_spectrip}, we have verified that, in the smooth setting,  
the conditions of the general index formulas 
(Theorem \ref{th-genericindexformula1} and \ref{th-genericindexformula2}) are met by the family 
$\{\lambda_d^\calL\}_{\calL \in \Xi}$ of spectral triples. Hence, our computations of the index pairings  
$\langle [p],[\lambda_d^\calL]\rangle$ (in even dimensions) and 
$\langle [u],[\lambda_d^\calL]\rangle$ (in odd dimensions), see Definition \ref{defn-indexpairing}, 
can proceed from \eqref{eq-genindexfeven} and \eqref{eq-genindexfodd}, respectively. 
Our main tool for evaluating these formulas is the following result, whose proof can be found in 
Section \ref{subsec:app_restrace} in the appendix.

\begin{lemma}\label{lem-reztrace} \label{lemma:d_dim_trace_per_unit_volume_related_to_residue_trace}
Let $f\in\calA_\mathrm{Sob}$. Then, $\bP$-almost sure, 
$$  
 \calT(f)=\Tr_\mathrm{Vol}(\wt{\pi}_\calL(f)) = \frac{1}{\mathrm{Vol}_{d-1}(S^{d-1})} 
  \res_{s=d}\Tr\!\left(\wt{\pi}_\calL(f)(1+|X|^2)^{-s/2}\right). 
$$
\end{lemma} 

We now state one of our main results.

\begin{thm}[Even formula]  \label{thm:complex_bulk_pairing_even}
Let $p$ be a  projection in $M_N(\calA)$ and 
suppose $d$ is even. Then the pairing of $p$ with the smooth spectral 
triple can $\bP$-almost surely be computed by the formula
\begin{align*}
  \Index\big( \pi_\calL(p)( F_X \otimes 1_N)_+ \pi_\calL(p) \big) & = C_{d}   \sum_{\rho\in S_d} 
  (-1)^\rho\, (\Tr_{\C^N}\otimes \calT) \bigg( p \prod_{j=1}^d \partial_{\rho(j)} p \bigg) \\
\nonumber  & = C_{d}   \sum_{\rho\in S_d} 
  (-1)^\rho\, (\Tr_{\C^N}\otimes \Tr_\mathrm{Vol}) \bigg( \pi_\calL(p) \prod_{j=1}^d [X_{\rho(j)},  \pi_\calL(p)] \bigg), 
\end{align*}
with $C_{2n} = \frac{(-2\pi i)^n}{n!}$, $\Tr_{\C^N}$ the matrix trace on $\C^N$ 
and $S_d$ the permutation group on $d$ letters. The formula is $\bP$-almost surely 
constant for any choice of $\calL\in\Xi$.
\end{thm}

\begin{proof} 
We will omit some of the details of the proof since, with Lemma \ref{lem-reztrace} in place, 
the arguments (in the even and odd setting) are exactly the same as in~\cite{BSBWeak}.
We consider the case of $N=1$ as case of general matrices is a simple extension. From \eqref{eq-genindexfeven}, 
$$
  \Index \big ( \pi_\calL(p) (F_X)_+ \pi_\calL(p) \big ) =   \res_{r=(1-d)/2} \Big( \sum_{m=1, \mathrm{even}}^{2N} 
     \phi^r_m ( \mathrm{Ch}_m(p) ) \Big)
$$
Because our space is flat 
and the Dirac operator globally defined, algebraic manipulation of the Dirac operator and the  corresponding Clifford generators means that 
only the top degree term in the local index formula will have non-zero residue as in~\cite[Appendix]{BCPRSW}. 
Hence the formula reduces to the residue of $\phi_d^r(\mathrm{Ch}^d(p))$.
To take the contour integral in $\phi_d^r(\mathrm{Ch}^d(p))$, we move all the resolvent terms to 
the right, which can be done up to a holomorphic correction. 
We can then compute the Cauchy integral and write the result in the form
$$
   \Index \big ( \pi_\calL(p) (F_X)_+ \pi_\calL(p) \big )  =  (-1)^{d/2}  \tfrac{1}{d} \res_{z=d} \,
     (\Tr\otimes \Tr_{\C^\nu}) \big(\Gamma_0 (2{p}-1)([X,{p}])^d(1+X^2)^{-z/2} \big), 
$$
where we recall $\Gamma_0 = (-i)^{d/2}\Gamma^1\cdots\Gamma^d$. There is a 
symmetry of the eigenspaces of $\Gamma_0$ that implies that the trace of
$\Gamma_0 ([X,{p}])^d(1+X^2)^{-z/2}$ will be holomorphic at $\Re(z)=d$ and so does  
not contribute to the index pairing. Writing the power $([X,{p}])^d$ in terms of permutations,  
and applying Lemma \ref{lemma:d_dim_trace_per_unit_volume_related_to_residue_trace} 
with the extra spinor degrees of freedom gives the result. 
We again refer the reader to~\cite{BSBWeak} for the complete algebraic details. 
The index formula is almost sure constant as for the ergodic measure, the 
spectral triples $\lambda_d^\calL$ have the same index pairing. Similarly, 
Lemma \ref{lemma:d_dim_trace_per_unit_volume_related_to_residue_trace} 
ensures that the residue trace is almost surely constant in $\calL$.
\end{proof}

\begin{thm}[Odd formula]  \label{thm:complex_bulk_pairing_odd}
Let $u$ be a complex unitary in $M_N(\calA)$ and and suppose 
the dimension $d$ is odd. 
Then the pairing of $u$ with the smooth spectral triple can $\bP$-almost surely be expressed by 
the formula
\begin{align*}
 \Index\big(\Pi_N \pi_\calL(u) \Pi_N - (1-\Pi_N)\big) &= \tilde{C}_{d} \sum_{\rho\in S_d}(-1)^\rho\,   (\Tr_{\C^N}\otimes \calT) 
 \bigg( \prod_{j=1}^d u^* \, \partial_{\rho(j)}u  \bigg)  \\
 &= \tilde{C}_{d} \sum_{\rho\in S_d}(-1)^\rho\, 
  (\Tr_{\C^N}\otimes \Tr_\mathrm{Vol}) 
    \bigg(\prod_{j=1}^d \pi_\calL(u)^* [X_{\rho(j)}, \pi_\calL(u)] \bigg), 
\end{align*}
where $\tilde{C}_{2n+1} = \frac{ 2(2\pi i)^n n!}{(2n+1)!}$, $\Tr_{\C^N}$ is the matrix trace on $\C^N$ 
and $S_d$ is the permutation group on $d$ letters. The formula is $\bP$-almost surely 
constant for any choice of $\calL\in\Xi$.
\end{thm}
\begin{proof}
As in the even case, only the top term contributes to the index pairing and so 
$$
   \Index\big(\Pi_q u \Pi_q - (1-\Pi_q)\big) =  \frac{-1}{\sqrt{2\pi i}} \res_{r=(1-d)/2} \phi_d^r(\mathrm{Ch}^d(u)).
$$
As before, we take the Cauchy integral and after some rearranging  
\begin{align*}
     &\Index\big(\Pi_q u \Pi_q - (1-\Pi_q)\big) \\
     &\hspace{2cm}  = (-1)^{n+1}\  \frac{n! \Gamma(k/2)}{d!\sqrt{\pi}}  \res_{z=d} \,
   (\Tr \otimes \Tr_{\C^\nu}) \big(u^*[X,u][X,u^*]\cdots[X,u](1+X^2)^{-z/2}\big).
\end{align*}
for $d=2n+1$ and with a product of $d$ commutators in the trace on the right-hand side. 
We use the identity $[X,u^*]=-u^*[X,u]u^*$, which implies 
$$
  u^*[X,u][X,u^*]\cdots[X,u] = (-1)^n \big( u^*[X,u]\big)^d.
$$
We express this power using permutations and compute the spinor and residue trace, where 
Lemma \ref{lemma:d_dim_trace_per_unit_volume_related_to_residue_trace} then 
gives the result.
\end{proof}

\subsection{The mobility gap regime} \label{subsec:mobility_gap_index}

The definition below is the operator theoretic formulation of a mobility gap, which is usually 
done using representations on Hilbert spaces. The latter will be difficult in the present general 
context because the Hilbert spaces of the representations change from one configuration to another.

\begin{defn}[Mobility gap]
Let $h \in M_N(C^\ast(\calG,\sigma))$ be self-adjoint. 
We call an interval $\Delta \subseteq {\rm Spec}(h)$ a mobility 
gap of $h$ if we have a \emph{continuous} morphism:
\begin{equation}\label{eq-mobgap}
L^\infty(\Delta \subset \R) \ni \varphi \mapsto 
   \varphi(h) \in M_N(\calA_\mathrm{Sob}).
\end{equation}
\end{defn}

\begin{remark} According to the above definition, the Fermi projector $p_F = \chi_{(-\infty,E_F]}$ 
of an electronic system does belong to the Sobolev algebra if $E_F$ resides in a mobility gap. 
This automatically implies that Anderson's localization length is finite and, furthermore, that all 
linear and non-linear direct transport coefficients vanish as the temperature goes to zero \cite{PB2016,ProdanSpringer2017}. 
Hence, if we are talking about electronic systems, \eqref{eq-mobgap} ensures that the systems are insulators.
We note that
 a proof of \eqref{eq-mobgap} for disordered crystals can be found in \cite{ProdanSpringer2017}. 
 It relies on the Aizenman-Molchanov bound \cite{AM93}.$\, \Diamond$
\end{remark}

\begin{lemma}[\cite{CPRS1}, Theorem 10] \label{lemma:sobolev_cocycle}
For $\bP$-almost all $\calL\in\Xi$, the multilinear functional
\begin{align} \label{eq:Sobolev_Hochschild}
   \phi(a_0,\ldots,a_d) &= \res_{z=d} \, (\Tr_{\C^\nu} \otimes \Tr) \big( \Gamma_0\, \wt{\pi}_\calL(a_0) [X,\wt{\pi}_\calL(a_1)]\cdots [X,\wt{\pi}_\calL(a_d)](1+X^2)^{-z/2} \big) 
\end{align}
is well-defined and continuous in the topology of $\calA_\mathrm{Sob}$ (where $\Gamma_0=1_\nu$ if $d$ is odd).
\end{lemma}
The functional $\phi$ is actually the Hochschild cocycle associated to the 
family $\{\wt \lambda_d^{\calL}\}_{\calL \in \Xi}$.  For more details, the reader can consult~\cite[Ch.~IV.2]{Connes94} 
or~\cite{CPRS1}. 
In particular, we can compute the residue trace in Equation \eqref{eq:Sobolev_Hochschild} and, 
applying some algebraic manipulation and the  
spinor trace, the functional $\bP$-almost surely reduces to 
$$
  \phi(a_0,\ldots,a_d) = C_d \sum_{\rho\in S_d} (-1)^\rho\, \calT\big( a_0 \partial_{\rho(1)}a_1 \cdots \partial_{\rho(d)}a_d \big),
$$
which is, again, continuous over $\calA_\mathrm{Sob}$.

\begin{thm}[Even formula]  \label{prop:mobilityindex_even}
Let $h \in M_N(\calA)$ and $[a,b]\subset \R$ be an interval with the ends in mobility gaps of $h$. 
Then the integer index pairing defined in Proposition \ref{prop-sobindexpairing} 
and applied to the spectral projector $p=\chi_{[a,b]}(h) \in M_N(\calA_\mathrm{Sob})$ accepts the local formula
$$
\Index\big( \wt \pi_\calL(p)( F_X \otimes 1_N)_+ \wt \pi_\calL(p) \big) = 
   C_{d}  \sum_{\rho\in S_d} (-1)^\rho \, (\Tr_{\C^N} \otimes \calT) 
   \bigg( p \prod_{j=1}^d \partial_{\rho(j)} p \bigg).
$$
where the equality holds $\bP$-almost surely.
\end{thm}

\begin{proof} 
The H\"{older} inequality of Equation \eqref{eq-holder} or Lemma \ref{lemma:sobolev_cocycle} 
ensures that the local formula for the index is a continuous functional over $M_N(\calA_{\rm Sob})$. 
For the left hand side, we recall that the $\bP$-almost sure family of spectral 
triples $\{\wt \lambda_d^{\calL}\}_{\calL \in \Xi}$ is $d$-summable. 
This automatically implies that the operator inside the index satisfies, $\bP$-almost surely, 
the Calderon-Fedosov principle from Proposition \ref{prop-fedosov}. This follows from a simple functional 
analytic argument (see \cite[Prop.~5.9]{CPRNotes}). Then the index can be $\bP$-almost 
surely expressed via the Connes--Chern character:
$$
  \Index\big( \wt \pi_\calL(p)( F_X \otimes 1_N)_+ \wt \pi_\calL(p) \big)  =
\tfrac{1}{2}   \Lambda_d (\Tr \otimes \Tr_{\C^\nu})\Big(\Gamma_0\, \wt{\pi}_\calL(p) \prod_{i=0}^d [F_X,\wt\pi_\calL(p)] \Big)
$$
with $\Lambda_d$ a constant (see~\cite[p295-296]{Connes94}). 
In the smooth case, because the 
top term in the local index computation survives, the Connes--Chern character can 
$\bP$-almost surely be 
computed using the functional $\phi(p,\ldots,p)$ from Lemma \ref{lemma:sobolev_cocycle} 
over $\calA$. But Lemma \ref{lemma:sobolev_cocycle} also shows that 
$\phi(p,\ldots,p)$ extends to $\calA_\mathrm{Sob}$ continuously 
(also see~\cite[Ch.~IV.2.$\gamma$, Theorem 8]{Connes94} or~\cite[Theorem 10]{CPRS1}).
Because the index formula holds on the dense subalgebra $\calA$ and both sides 
can be continuously extended over $\calA_\mathrm{Sob}$, the index formula extends.
\end{proof}

The method of proof used in Theorem \ref{prop:mobilityindex_even} can also be 
applied to the odd index pairing.

\begin{thm}[Odd formula] \label{prop:mobilityindex_odd}
Suppose $d$ is odd and $h \in M_{2N}(\calA)$ is self-adjoint and 
chiral symmetric, {\it i.e.} 
$$
h\begin{pmatrix} 1_N & 0 \\ 0 & -1_N \end{pmatrix}=-\begin{pmatrix} 1_N & 0 \\ 0 & -1_N \end{pmatrix}h.
$$
Let $[a,b] \subset (-\infty,0]$ be an interval with ends residing in mobility gaps of $h$ 
and $p=\chi_{[a,b]}(h)$ be the associated spectral projection. 
Let $u \in M_N(\calA_{\rm Sob})$ be the unitary element appearing in the decomposition 
$1-2p=\begin{pmatrix} 0 & u^\ast \\ u & 0 \end{pmatrix}$. Then the integer index pairing defined in 
Proposition \ref{prop-sobindexpairing} and applied on $u$ accepts the local formula
\begin{equation}
\Index\big( \Pi_N \wt{\pi}_\calL(u) \Pi_N - (1-\Pi_N) \big) 
   = \tilde{C}_{d} \sum_{\rho\in S_d} (-1)^\rho\, 
  (\Tr_{\C^N} \otimes \calT) \bigg( \prod_{j=1}^d u^*\,\partial_{\rho(j)}u\bigg),
 \end{equation}
 where the equality holds $\bP$-almost surely.
\end{thm}

\section{Discussion and conclusions}
\label{Sec-Discussions}

We now return to the patterned resonators introduced in Section~\ref{Sec:Patterned_resonators} 
and examine them with the formalism introduced in Section~\ref{Sec-GrupoidAlg} and the results 
of Section~\ref{Sec-IndexFormulas}. In particular, we consider how to connect practical situations 
with our mathematical formalism and to spell out the conditions in which the quantization and stability of the Chern number holds.

\vspace{0.2cm}

We recall that there are algorithms that produce an entire pattern directly in the thermodynamic limit, 
such as the dynamically generated patterns \cite{HPS2017} or the model sets \cite{SadunBook}. 
Other algorithms produce families of patterns, such as the one used to produce the amorphous 
pattern in Section~\ref{Sec-QHall}, which can only be defined as thermodynamic limits of finite patterns. 
Families of patterns can also come from the pure thermodynamic phases of the condensed matter. 
They can all be characterized by an ergodic dynamical system $(\Omega,\R^d,T,{\rm d}\bP)$ as 
explained in Section~\ref{Sec-GrupoidAlg}. 

\vspace{0.2cm}

The dynamical systems from dynamically generated patterns and models sets are 
topologically minimal and uniquely ergodic, hence ${\rm d}\bP$ is automatically determined by the pattern and no 
additional data is needed beyond the topological dynamical system $(\Omega,\R^d,T)$. 
For patterns arising as themodynamic limits of finite patters, 
there are many ergodic measures available on $\Omega$. In such cases, the algorithm itself 
produces the the probability measure, as we've seen in Section~\ref{Sec-QHall}. Ergodicity of this measure, 
which is ultimately a property of the algorithm, is a key assumption for the results in Section~\ref{Sec-IndexFormulas}. 
For the patterns associated with condensed matter systems, this means the Gibbs measure for the 
atomic degrees of freedom must be ergodic or, in other words, the results in Section~\ref{Sec-QHall} 
apply only to the thermodynamically pure phases. To make sure our statement is understood correctly, 
let us point out that the assumptions in  Section~\ref{Sec-QHall} are optimal as, otherwise, we can easily 
produce counter examples. Hence, the quantization of the invariants will generically fail beyond the precise 
conditions stated in Section~\ref{Sec-QHall}.

\vspace{0.2cm}

Regarding the dynamics of the coupled resonant modes, we introduced in 
Section~\ref{Sec:Patterned_resonators} the very physical assumptions that the couplings 
between the pattered resonators are uniquely determined by the configuration of the points. Namely,  
the couplings  become irrelevant beyond some large but finite range and the hopping matrices 
depend continuously on the pattern (viewed as point in the space of Delone sets).  
In this very physical setting, we treated the hopping matrices as continuous functions over $\Omega$ 
and this revealed that all the Hamiltonians $H_\calL$ driving the dynamics of any such coupled 
resonators possess a certain structure. In particular, they all can be generated, via canonical 
representations, from the smooth subalgebra $\calA$ of the groupoid algebra. Given a pattern 
or a family of patterns, we described in Section~\ref{Sec-GrupoidAlg} how to navigate from 
the physical representation to the algebraic one and back. This is important for the practical 
aspects too, because the numerical codes used in Section~\ref{Sec-QHall} were generated in the algebraic framework.

\vspace{0.2cm}

An interesting question about patterned resonators we wanted to answer in our work is how to 
detect if different parts of the resonant spectrum carry non-trivial topological invariants. If the 
spectral region is isolated, {\it i.e.} is flanked by spectral gaps, then the spectral projection 
$p = \chi_{[a,b]}(h)$ belongs to the smooth algebra and the 
Theorems~\ref{thm:complex_bulk_pairing_even} and \ref{thm:complex_bulk_pairing_odd} provide 
the answer. We should mention that these cases are quite relevant for systems engineered 
with meta-materials like photonic and accoustic crystals where the disorder can be kept under 
control and various parts of the spectrum can be populated (excited or pumped) at will. 
This, however, is not the case for electronic systems where thermal disorder is unavoidable 
and large at room temperature (where topological insulators are supposed to function) 
and the electrons populate the spectrum according to Pauli's principle. In this latter 
case (but not exclusively), the regime where the spectral region is actually 
flanked by mobility gaps is much more relevant.

\vspace{0.2cm}

Results concerning the spectral properties and decomposition of Hamiltonians in the general Delone setting are still 
in development, see~\cite{LPV07, GMRM15, RMreview} for a more detailed overview. 
The major difference when compared to the case of disordered crystals, which is quite well understood, 
is that the Hamiltonians $H_\calL$ act on different Hilbert spaces. Elucidating the spectral characteristics 
of these class of Hamiltonians was  beyond the scope of our study. We opted instead on formulating an 
operator theoretic definition of a mobility gap, which is correct from the physical point of view and captures 
most which is known about the disordered crystals. Furthermore, it appears to us that its proofs 
from \cite{ProdanSpringer2017} for disordered crystals can be adapted to this more general context, 
at least for the systems which accept periodic approximates. 

\vspace{0.2cm}

Per the above discussion, the regime of mobility gaps is covered by Theorems~\ref{prop:mobilityindex_even} 
and \ref{prop:mobilityindex_odd}. They confirm the quantization and the stability of the Chern numbers 
in this regime. Indeed, a deformation $h_t$ inside the smooth algebra of the Hamiltonian leads 
to a homotopy of spectral projections $p_t$ in $\calA_{\rm Sob}$. Then, the local expression 
of the index formula is a continuous functional over $\calA_\mathrm{Sob}$, while the equality 
with a Fredholm index pins the range of the Chern numbers to integers. 
As such, the only way to change the value of the Chern number of a spectral projection is to force 
the ends of the spectral interval leave the mobility gap ({\it i.e.} pass through a region of 
extended states).

\section{Appendix: Background on non-commutative index theory}

Here we give a brief overview of index theory of $C^*$-algebras. Further details and proofs 
of the results can be found in~\cite{BBB, CPRNotes, ElementsNCG, HigsonRoe}.

\subsection{Fredholm index}

\begin{defn}
Let $\calH_1, \calH_2$ be Hilbert spaces and $F:\calH_1\to \calH_2$ a bounded linear operator. 
We say that $F$ is Fredholm if
\begin{enumerate}
  \item $\Ran(F)$ is closed in $\calH_2$,
  \item $\Ker(F)$ and $\coKer(F) = \calH_2/\Ran(F)$ is finite dimensional.
\end{enumerate}
If $F$ is Fredholm we define
$$
  \Index(F) = \mathrm{dim}\Ker(F) - \mathrm{dim} \coKer(F).
$$
\end{defn} 

While Fredholm operators come from a purely analytic definition they also have topological properties.

\begin{thm} \label{thm:Fred_props}
Let $\calF$ denote the set of Fredholm operators on a fixed Hilbert space $\calH$, and let 
$\pi_0(\calF)$ denote the set of (norm) connected components of $\calF$.
\begin{enumerate}
  \item  If there are operators $F,S\in \calB(\calH)$ such that $1-FS,\, 1-SF$ are compact, then $F,S\in\calF$.
  \item For $F,S\in\calF$, 
 $$
 \Index(F^*) = - \Index(F ), \qquad  \Index(F S) = \Index(F ) + \Index(S).
 $$
   \item The index is locally constant on $\calF$ (in the operator norm) and induces a group isomorphism $\Index : \pi_0(\calF) \to \Z$.
  \item If $F$ is Fredholm and $K$ is compact then $F + K$ is Fredholm and $\Index(F+K) = \Index(F)$. 
\end{enumerate}
\end{thm}

\begin{prop}[Fedosov--Calderon principle~\cite{Cal,Fed}]\label{prop-fedosov} An operator $F \in \calB(\calH)$ 
with $\|F\|\leq 1$ is Fredholm provided there is a positive integer $n$ such that $(1-F F^\ast)^n$ and 
$(1-F^\ast F)^n$ are trace class. Furthermore, if this is the case, the Fredholm index can be computed as
$$
\Index(F) = \Tr (1-F F^\ast)^n - (1-F^\ast F)^n.
$$ 
\end{prop}

In classical index theory on manifolds, we are often interested in the Fredholm index of operators that are 
derived from elliptic differential operators. The analogue of this structure for $C^*$-algebras (with 
a dense subalgebra) is a spectral 
triple. 

\subsection{Spectral triples}

In what follows we will assume that all algebras we work with are unital.

\begin{defn}\label{defn-spectriple}
A spectral triple $(\calA,{}_{\pi}\calH,D)$ is given by a unital $\ast$-algebra $\calA$ with representation 
$\pi:\calA \to \calB(\calH)$ with $\pi(1_\calA) = 1_\calH$
and a densely defined self-adjoint (typically unbounded) operator $D:\Dom(D)\subset\calH \to \calH$
satisfying the following conditions. 
\begin{enumerate}
  \item For all $a\in\calA$, $\pi(a)\Dom(D)\subset \Dom(D)$ and the densely defined operator 
  $[D, \pi(a)] := D\pi(a) - \pi(a)D$ is bounded and so it extends to a bounded operator on all of $\calH$ by continuity. 
  \item The operator $(1+D^2)^{-1/2}$ is compact.
\end{enumerate}
If in addition there is an operator $\gamma\in\calB(\calH)$ with $\gamma^* = \gamma$, $\gamma^2=1$, 
$D\gamma + \gamma D=0$ and $[\pi(a),\gamma]=0$ for all $a\in\calA$, 
 we call the spectral triple even or graded. Otherwise it is odd or ungraded.
\end{defn}

\begin{prop}[\cite{BJ83}]
If $(\calA,{}_{\pi}\calH,D)$ is a spectral triple, then:
\begin{enumerate} 
\item $F_D=D(1+D^2)^{-1/2}$ is a Fredholm operator 
and $[F_D,\pi(a)]$ is compact for all $a \in \calA$. 
\item $(\calA,{}_{\pi}\calH,F_D)$ determines a class in the $K$-homology group $K^\ast(A)$, where $A$ is the 
$C^*$-closure of $\calA$ and $\ast=0,1$ depending on whether the spectral triple is even or odd.
\end{enumerate}
\end{prop}

The following two properties define the class of spectral triples where a local index formula for the 
pairings in Definition \ref{defn-indexpairing} is available (see further below).

\begin{defn}[QC-regularity]\label{defn-qccondition}
A spectral triple $(\calA,{}_{\pi}\calH,D)$ is $QC^k$ for $k\geq 1$ ($Q$ for quantum) if for all $a\in\calA$ 
the operators $\pi(a)$ and $[D,\pi(a)]$ are in the domain of $\delta^k$ where $\delta(T) = [|D|,T]$ is the partial 
derivation on $\calB(\calH)$ defined by $|D|$. We say that $(\calA,{}_{\pi}\calH,D)$ is $QC^\infty$ if it is $QC^k$ 
for all $k\geq 1$.
\end{defn}

\begin{defn}[Summability]\label{defn-summability}
A spectral triple $(\calA,{}_{\pi}\calH,D)$ is finitely summable if there is some $s>0$ such that 
$$
  \Tr((1+D^2)^{-s/2}) < \infty.
$$
The infimum $s_0$ of all such $s$ is called the spectral dimension and we say that 
$(\calA,{}_{\pi}\calH,D)$ is $s_0$-summable.
\end{defn}

\begin{prop}\label{prop-summability} 
Let $(\calA,{}_{\pi}\calH,D)$ be a $p$-summable and $QC^1$ 
spectral triple with $p \geq 0$. 
Then for $F_D = D(1+D^2)^{-1/2}$, $[F_D,\pi(a)]$ 
belongs to the $(\lfloor p\rfloor+1)$-Schatten 
class for all $a\in \calA$.
\end{prop}

\begin{example}
Consider the $n$-dimensional torus $\T^n$, which has the trivial 
complex spinor bundle $S\to \T^n$.
The smooth algebra $C^\infty(\T^n)$ acts diagonaly on the space of spinors 
$L^2(\T^n, S)$ by left-multiplication $(M(f)\psi)(k)=f(k)\psi(k)$. 
The spinor bundle comes with the self-adjoint matrices $\{\gamma^j\}_{j=1}^n$ 
such that $\gamma^i \gamma^j + \gamma^j \gamma^i = 2\delta_{i,j}$. 
When $n$ is even, $\gamma_0 = (-i)^{n/2} \gamma^1\cdots \gamma^n$ 
gives a splitting (grading) $S\cong S_+\oplus S_-$. 
Defining the self-adjoint Dirac operator 
$\sD = -i\sum_{j=1}^n \gamma^j \partial_{j}$, we claim that 
that 
$(C^\infty(\T^n), L^2(\T^n,S), \sD \big)$ is a spectral triple. 
We first check that 
$$
  [\sD,M(f)]\psi = -i\sum_{j=1}^n \gamma^j \big( \partial_{j}(f\psi) - f\partial_{j}(\psi) \big)
   = -i\sum_{j=1}^n M(\partial_{j} f)\psi,
$$
which is clearly bounded for $f\in C^\infty(\T^n)$. The 
ellipticity of $\sD$ implies that $(1+\sD^2)^{-1/2}$ is 
compact. When $n$ 
is even, $\sD$ anti-commutes with $\gamma_0$ and we 
have a spectral triple that is even or odd depending on 
the parity of $n$.  Because 
$C^\infty(\T^n)$ is invariant under differentiation, the triple 
is $QC^\infty$. Standard Fourier analysis can be used to 
show that $(1+\sD^2)^{-s/2}$ is trace-class for $s>n$ and, 
hence, the spectral triple is $n$-summable. 
$\, \Diamond$
\end{example}

Motivated by our application to covariant observables with a mobility gap, we consider a slight 
extension of a spectral triple.

\begin{defn}\label{defn-sobspectriple}
A family $(\calA,{}_{\pi_\xi}\calH_\xi,D_\xi)$ indexed by $\xi$ from a probability space $(\Xi,\bP)$ is called a 
$\bP$-almost sure family of spectral triples if $\calA$ is a unital $\ast$-algebra with a family of representations 
$\pi_\xi:\calA \to \calB(\calH_\xi)$ with $\pi_\xi (1_\calA) = 1_{\calH_\xi}$
and $\{D_\xi\}_{\xi\in\Xi}$ is a family of densely-defined self-adjoint 
operators $D_\xi:\Dom(D_\xi)\subset\calH_\xi \to \calH_\xi$ that
satisfy the following conditions:
\begin{enumerate}
  \item For any fixed but arbitrary $a\in\calA$, $\bP$-almost sure $\pi_\xi(a)\Dom(D_\xi)\subset \Dom(D_\xi)$ and the densely defined operator 
  $[D_\xi, \pi_\xi(a)] := D_\xi\pi_\xi(a) - \pi_\xi(a)D_\xi$ is $\bP$-almost sure bounded and so it extends to a bounded operator on all of $\calH_\xi$ by continuity. 
  \item The operator $(1+D_\xi^2)^{-1/2}$ is compact for all $\xi\in\Xi$.
\end{enumerate}
If in addition there is family of operators $\gamma_\xi\in\calB(\calH_\xi)$ with $\gamma_\xi^* = \gamma_\xi$, 
$\gamma_\xi^2=1$, 
$D_\xi\gamma_\xi + \gamma_\xi D_\xi=0$ and $[\pi_\xi(a),\gamma_\xi]=0$ for all $a\in\calA$, 
 we call the spectral triple even or graded. Otherwise it is odd or ungraded.
\end{defn}

There is a close connection between an almost sure family of spectral triples and so-called 
semifinite spectral triples, the interested reader can see~\cite{CPRS2,CPRS3,PRS} for more information. 

\subsection{The index pairing}

Recall that a spectral triple $(\calA,{}_\pi \calH,D)$ gives a  
 $K$-homology class. The $K$-homology group is the algebraic dual of the $K$-theory 
group and, as such, can be used to give numerical invariants to $K$-theory classes. 
In particular, our aim is to use the Fredholm operator $F_D = D(1+D^2)^{-1/2}$ to define 
an additive integer-valued pairing of a spectral triple with a $K_\ast$-theory class.

\begin{defn}[The index pairing, see~\cite{HigsonRoe}, Sec.~8.7]\label{defn-indexpairing} 
The following statements do not require any $QC$-regularity or summability.
\begin{enumerate}

\item Let $(\calA,{}_{\pi}\calH,D)$ be an even spectral triple and suppose $p\in M_k(\calA)$ is a projection. 
Then the off-diagonal part $( F_D \otimes 1_k)_+: (\calH^{\oplus k})_+ \to (\calH^{\oplus k})_-$, 
coming from the grading of the spectral triple, is Fredholm and
for $[p]\in  K_0(\calA)$ and $[(\calA,{}_\pi \calH,F_D)]\in K^0(\calA)$, the pairing 
\begin{equation}\label{eq-pairing1}
  \langle [p], [(\calA,{}_{\pi}\calH,D)] \rangle = \Index\big( \pi(p)( F_D \otimes 1_k)_+ \pi(p)\big), 
\end{equation}
gives a well-defined group homomorphism $K_0(\calA) \to \Z$.

\item Let $(\calA,{}_\pi \calH,D)$ be an odd spectral triple and suppose $u\in M_k(\calA)$ is unitary. 
Then $\Pi_k u \Pi_k - (1-\Pi_k)  : \calH^{\oplus k} \rightarrow \calH^{\oplus k}$ with $\Pi_k = \frac{1}{2}(1+F_D)\otimes 1_k$ is 
Fredholm and, for $[u]\in K_1(\calA)$ and $[(\calA,{}_\pi \calH,F_D)]\in K^1(\calA)$, the pairing
\begin{equation}\label{eq-pairing2}
  \langle [u], [(\calA,{}_{\pi}\calH,D)] \rangle = \Index\big( \Pi_k \pi(u) \Pi_k - (1-\Pi_k) \big),
\end{equation}
gives a well-defined group homomorphism $K_1(\calA)\to \Z$.
\end{enumerate}
\end{defn}

We note that if $\Pi_k = \frac{1}{2}(1+F_D)\otimes 1_k$ is a projection, {\it e.g.} for $D$ invertible 
and $F_D = D|D|^{-1}$, then the odd pairing simplifies to $\Index(\Pi_k u\Pi_k)$.

\begin{remarks}
\begin{enumerate}
  \item To prove that the operators $\pi(p)( F_D \otimes 1_k)_+ \pi(p)$ 
  and $\Pi_k \pi(u) \Pi_k - (1-\Pi_k)$ are Fredholm, one shows that 
  $\pi(p)( F_D \otimes 1_k)_{-} \pi(p)$ and $\Pi_k \pi(u)^* \Pi_k - (1-\Pi_k)$ provide inverses modulo compacts. 
  Then we can apply part (i) of Theorem \ref{thm:Fred_props}.
  \item Because the index pairing is a pairing of $K$-theory and $K$-homology, the index will not change 
if we replace a spectral triple with another spectral triple that defines the same $K$-homology class. 
Similarly, we can take homotopies on the $K$-theory side and the pairing remains constant. 
$\, \Diamond$
\end{enumerate}
\end{remarks}

\subsection{The local index formula} \label{subsec_localindexappendix}
The index pairing of a $K$-theory class with a spectral triple is given by an analytic formula that 
is explicit but is in general quite difficult to compute. Much like the classical Atiyah--Singer Index 
Theorem, it would be advantageous to describe this index pairing via a local formula more 
amenable to computation and numerical simulation. 
Such formula is indeed available for $QC^\infty$ and finitely summable spectral triples, 
which expresses the index pairing in terms of derivations and traces on the algebra $\calA$. 
The formula also involves a residue and so we can make pertubrations and corrections up to a 
holomorphic function in the neighbourhood of the residue. 

\vspace{0.2cm}

To state the index formula, we introduce some notation. Given $u\in \calA$ unitary, $p\in\calA$ 
a self-adjoint projection and $m\in \N$
\begin{align*}
      \mathrm{Ch}_{2m}(p) &= (-1)^m \frac{(2m)!}{2(m!)} (2p-1) \otimes p^{\otimes 2m}, \qquad \mathrm{Ch}_0(p) = p, \\
   \mathrm{Ch}_{2m+1}(u) &= (-1)^m m! \, u^* \otimes u \otimes u^* \otimes \cdots \otimes u, 
       \quad (2m+2 \text{ entries}).
\end{align*}
Our notation comes from the theory of cyclic homology and cohomology. The interested reader can 
consult~\cite[Ch.~III, IV]{Connes94}. 

\begin{thm}[Even index formula, \cite{CoM,CPRS3}]\label{th-genericindexformula1}
Let $(\calA,\calH,D)$ be an even finitely summable and $QC^\infty$ spectral triple with spectral 
dimension $q\geq 1$ and grading $\gamma$. Let $N= \lfloor (q+1)/2 \rfloor$ and suppose $p\in \calA$ 
is a self-adjoint projection. Then 
\begin{equation}\label{eq-genindexfeven}
  \Index( p (F_D)_+ p ) =   \res_{r=(1-q)/2} \Big( \sum_{m=1, \mathrm{even}}^{2N} 
     \phi^r_m ( \mathrm{Ch}_m(p) ) \Big),
\end{equation}
where for $a_0,\ldots,a_m \in \calA$, $l=\{u+iv\,:\,v\in\R\}, 0 < u < 1/2$, 
$R_s(\lambda) = (\lambda-(1+s^2+D^2))^{-1}$ and $r>1/2$ we define 
$\phi_m^r(a_0,\ldots,a_m)$ as
$$
  \frac{(m/2)!}{m!}  \int_0^\infty \! 2^{m+1} s^m \Tr \Big( \gamma \frac{1}{2\pi i} \int_l 
     \lambda^{-q/2-r} a_0 R_s(\lambda) [D,a_1] R_s(\lambda) \cdots [D,a_m] R_s(\lambda) \,\mathrm{d}\lambda 
     \Big) \mathrm{d}s.
$$
In particular the sum on the right hand side of the index formula analytically continues to a 
deleted neighbourhood of $r = (1-q)/2$
with at worst a simple pole at $r = (1 - q)/2$.
\end{thm}

\begin{thm}[Odd index formula, \cite{CoM,CPRS2}]\label{th-genericindexformula2}
Let $(\calA,\calH,D)$ be an odd finitely summable and $QC^\infty$ spectral triple with spectral 
dimension $q\geq 1$. Let $N= \lfloor q/2 \rfloor+1$ and $u\in \calA$  
unitary. Then 
\begin{equation}\label{eq-genindexfodd}
  \Index( \Pi u \Pi -(1-\Pi) ) = \frac{-1}{\sqrt{2 \pi i}} \res_{r=(1-q)/2} \Big( \sum_{m=1, \mathrm{odd}}^{2N-1} 
     \phi^r_m ( \mathrm{Ch}_m(u) ) \Big),
\end{equation}
where for $a_0,\ldots,a_m \in \calA$, $l=\{u+iv\,:\,v\in\R\}, 0 < u < 1/2$, 
$R_s(\lambda) = (\lambda-(1+s^2+D^2))^{-1}$ and $r>0$ we define 
$\phi_m^r(a_0,\ldots,a_m)$ as
$$
  \frac{-2\sqrt{2\pi i}}{\Gamma((m+1)/2)} \int_0^\infty \! s^m \Tr \Big( \frac{1}{2\pi i} \int_l 
     \lambda^{-q/2-r} a_0 R_s(\lambda) [D,a_1] R_s(\lambda) \cdots [D,a_m] R_s(\lambda) \,\mathrm{d}\lambda 
     \Big) \mathrm{d}s.
$$
In particular the sum on the right hand side of the index formula analytically continues to a 
deleted neighbourhood of $r = (1-q)/2$
with at worst a simple pole at $r = (1 - q)/2$.
\end{thm}

\subsection{Evaluating residue traces} \label{subsec:app_restrace}
\label{lem-reztrace2}

Here we complete the proof of Lemma \ref{lemma:d_dim_trace_per_unit_volume_related_to_residue_trace}, 
which we restate for convenience.

\begin{lemma} \label{lemma:appendix_res_trace}
Let $f\in\calA_\mathrm{Sob}$. Then, $\bP$-almost sure, 
$$  
 \calT(f)=\Tr_\mathrm{Vol}(\wt{\pi}_\calL(f)) = \frac{1}{\mathrm{Vol}_{d-1}(S^{d-1})} 
  \res_{s=d}\Tr\!\left(\wt{\pi}_\calL(f)(1+|X|^2)^{-s/2}\right). 
$$
\end{lemma}
\begin{proof}
We recall that for $f\in\calA_\mathrm{Sob}$ the algebraic trace is given by
$$ 
\calT(f) = \int_{\Xi}f(\calL,0)\,\mathrm{d}\bP(\calL). 
$$
We know from Proposition \ref{prop:sobolev_spectrip} that $(1+|X|^2)^{-s/2}$ 
is a trace class operator on $\ell^2(\calL)$ for $\Re(s)>d$. Hence we can take the 
trace of $\wt{\pi}_\calL(f)(1+|X|^2)^{-s/2}$ 
by summing along the diagonal of the integral kernel, where
$$  
\Tr\!\left(\wt{\pi}_\calL(f)(1+|X|^2)^{-s/2}\right) = 
  \sum_{x\in\calL} f(\calL-x,0)(1+|x|^2)^{-s/2}. 
$$
We denote by $G(\calL,s) = \Tr\!\left(\wt{\pi}_\calL(f)(1+|X|^2)^{-s/2}\right)$ 
for $\Re(s)>d$. Suppose that there is an $a\in\R^d$ such that 
$\calL+a \in\Xi$, we compute that
\begin{align*}
  G(\calL+a,s) &= \sum_{x\in\calL+a} f(\calL+a-x,0)(1+|x|^2)^{-s/2} \\
   &=  \sum_{u\in\calL} f(\calL-u,0)(1+|a+u|^2)^{-s/2}  \\
   &=  \sum_{u\in\calL} f(\calL-u,0) (1+|u|^2)^{-s/2}  \\
   &\hspace{0.5cm}+  \sum_{u\in\calL} f(\calL-u,0)\!
     \left( (1+|a+u|^2)^{-s/2} - (1+|u|^2)^{-s/2}\right).
\end{align*}
Our aim is to show that the difference $G(\calL+a,s) - G(\calL,s)$ is holomorphic at 
$\Re(s)=d$ and so the residue will be constant on an orbit.
We use the Laplace transform to rewrite 
\begin{align*}
   G(\calL+a,s) - G(\calL,s) &= \sum_{u\in\calL} f(\calL-u,0)\!
     \left( (1+|a+u|^2)^{-s/2} - (1+|u|^2)^{-s/2}\right) \\
  &\hspace{-1cm}= \frac{1}{\Gamma\left(\frac{s}{2}\right)} \sum_{u\in\calL} 
    f(\calL-u,0)\int_0^\infty \!t^{s/2-1}\left(e^{-t(1+|a+u|^2)} - e^{-t(1+|u|^2)} \right)\mathrm{d}t \\
  &\hspace{-1cm}= \frac{1}{\Gamma\left(\frac{s}{2}\right)} \sum_{u\in\calL} 
    f(\calL-u,0)\int_0^\infty \!t^{s/2-1} \int_0^a \nabla_b  
       \left(e^{-t(1+|b+u|^2)}\right) \mathrm{d}b\,\mathrm{d}t. 
\end{align*}
Taking the derivative in $b$ we find that (using multi-index notation)

\begin{align*}
  G(\calL+a,s) - G(\calL,s)  &= \frac{1}{\Gamma\left(\frac{s}{2}\right)} \sum_{u\in\calL} 
    f(\calL-u,0)\int_0^\infty \!t^{s/2-1} \int_0^a (-2t|b+u|) e^{-t(1+|b+u|^2)}\, \mathrm{d}b\,\mathrm{d}t  \\
  &= \frac{1}{\Gamma\left(\frac{s}{2}\right)} \sum_{u\in\calL} 
   f(\calL-u,0)\int_0^\infty \!t^{s/2} \int_0^a (-2|b+u|) e^{-t(1+|b+u|^2)}\, 
      \mathrm{d}b\,\mathrm{d}t   \\
  &= \frac{\Gamma\left(\frac{s}{2}+1\right)}{\Gamma\left(\frac{s}{2}\right)} 
   \sum_{u\in\calL} f(\calL-u,0) \int_0^a (-2|b+u|)(1+|b+u|^2)^{-s/2-1}\, 
     \mathrm{d}b   \\
  &= -s \sum_{u\in\calL} f(\calL-u,0) \int_0^a |b+u|(1+|b+u|^2)^{-s/2-1}\,
      \mathrm{d}b.
\end{align*}
We note that the last sum will coverge for $\Re(s)>d-1$. The difference 
$G(\calL+a,s)-G(\calL,s)$ is holomorphic for $\Re(s)>d-1$. 
To prove this claim, we first compute 
\begin{align*}
  &\frac{1}{h}\left(G(\calL+a,s+h)-G(\calL,s+h)-G(\calL+a,s)+G(\calL,s)\right) \\ 
   &\hspace{-0.1cm}= -\sum_{u\in\calL} f(\calL-u,0) \int_0^a |b+u|\! 
    \left(\frac{(s+h)(1+|b+u|^2)^{-\frac{h}{2}} - s}{h}\right)\!(1+|b+u|^2)^{-\frac{s}{2}-1}\,
      \mathrm{d}b
\end{align*}
and compare to the formal derivative
$$ 
  -\sum_{u\in\calL}  f(\calL-u,0) \int_0^a |b+u|\!
  \left( 1-\frac{1}{2}\ln(1+|b+u|^2) \right)\!(1+|b+u|^2)^{-\frac{s}{2}-1}\,\mathrm{d}b
$$
whose integral will also converge for $\Re(s)>d-1$. We then check that
\begin{align*}
  \lim_{h\to 0}\frac{(s+h)(1+|b+u|^2)^{-\frac{h}{2}} - s}{h} 
  &= \lim_{h\to 0}\frac{(s+h)\exp\!\left(-\frac{h}{2}\ln(1+|b+u|^2)\right) -s}{h} \\
  &= \lim_{h\to 0} \frac{(s+h)\!\left(1-\frac{h}{2}\ln(1+|b+u|^2)+ \mathcal{O}(h^2)\right) -s}{h} \\
  &= 1-\frac{1}{2}\ln(1+|b+u|^2).
\end{align*}
Therefore $G(\calL+a,s)-G(\calL,s)$ has a well-defined complex derivative for $\Re(s)>d-1$.

Next we fix some $\calL_0 \in \Xi$. By the ergodicity hypothesis, 
we know that $\calL_0$ is $\bP$-almost surely in the same orbit of $\calL$. 
We consider the function 
$\calL \mapsto G(\calL,s)-G(\calL_0,s)$. Integrating yields
$$ 
\int_\Xi \!\big(G(\calL,s) - G(\calL_0,s)\big)\mathrm{d}\bP(\calL) 
= \int_\Xi G(\calL,s)\,\mathrm{d}\bP(\calL) - G(\calL_0,s) 
$$
as $\bP(\Xi)=1$. For $\Re(s)>d$,
\begin{align*}
  \int_\Xi G(\calL,s)\,\mathrm{d}\bP(\calL) 
   &= \int_\Xi \sum_{u\in\calL} f(\calL-u,0)(1+|u|^2)^{-s/2}\, \mathrm{d}\bP(\calL) \\
   &= \int_\Xi \sum_{u\in\calL} f(\calL,0)(1+|u|^2)^{-s/2}\,\mathrm{d}\bP(\calL)
\end{align*}
where we have used the invariance of the action on $\Xi$ to make a 
substitution. Because $\calL\in\Xi$ is $(r,R)$-Delone, we can approximate the 
sum $ \sum_{u\in\calL} (1+|u|^2)^{-s/2}$ by an integral in polar coordinates 
(where the approximation becomes exact in the residue). Hence 
we can explicitly compute 
\begin{align*}
  \int_\Xi G(\calL,s)\,\mathrm{d}\bP(\calL) 
   &= \mathrm{Vol}_{d-1}(S^{d-1})  \int_\Xi f(\calL,0)\,\mathrm{d}\bP(\calL) 
      \int_0^\infty (1+r^2)^{-s/2} r^{d-1}\,\mathrm{d}r + h(s) \\
  &= \calT(f)\,\mathrm{Vol}_{d-1}(S^{d-1}) 
     \frac{\Gamma\!\left(\frac{d}{2}\right) \Gamma\!\left(\frac{s-d}{2}\right)}{2\Gamma\!\left(\frac{s}{2}\right)} + h(s)
\end{align*}
with $h(s)$ a function holomorphic in a neigbourhood of $\Re(s)=d$. We remark that 
the $R$-relatively dense property of $\calL$ is needed here.

As $g(\calL,s) =G(\calL,s)-G(\calL_0,s)-h(s)$ is $\bP$-almost surely holomorphic in a neighbourhood of $s=d$, we can say that
\begin{equation} \label{eq:residue_integral}
 \calT(f)\,\mathrm{Vol}_{d-1}(S^{d-1})  
  \frac{\Gamma\!\left(\frac{d}{2}\right) \Gamma\!\left(\frac{s-d}{2}\right)}{2\Gamma\!\left(\frac{s}{2}\right)}  
    =  \int_\Xi g(\calL,s) \mathrm{d}\bP (\calL) + G(\calL_0,s)
\end{equation}
By the functional equation for the $\Gamma$-function, the left hand side 
of Equation \eqref{eq:residue_integral} has an analytic continuation to the complex plane 
with a simple pole at $s=d$. We also note that, by the 
invariance of the measure, $\int_\Xi g(\calL,s)\, \mathrm{d}\bP$ is $\bP$-almost surely holomorphic for 
$\Re(s)>d$ and $g$ 
is holomorphic function in a neighbourhood of $s=d$. Therefore we conclude that 
$G(\calL_0,s)$ analytically extends to a neighbourhood of $s=d$ such 
that $(s-d)G(\calL_0,s)$ is holomorphic at $s=d$ for $\bP$-almost any $\calL_0\in\Xi$. 
Computing the residue,
\begin{align*}
 \res_{s=d} \Tr\!\left(\wt{\pi}_{\calL_0}(f)(1+|X|^2)^{-s/2}\right) 
   &= \res_{s=d} \calT(f)\,\mathrm{Vol}_{d-1}(S^{d-1}) 
     \frac{\Gamma\!\left(\frac{d}{2}\right) \Gamma\!\left(\frac{s-d}{2}\right)}{2\Gamma\!\left(\frac{s}{2}\right)} \\
   &= \calT(f)\,\mathrm{Vol}_{d-1}(S^{d-1}).
\end{align*}
Lastly, we use the canonical extension of Proposition \ref{prop:ergodic_trace_is_vol_trace} to $\calA_\mathrm{Sob}$ 
and conclude that, $\bP$-almost surely,
\begin{equation*}
     \res_{s=d} \Tr\!\left(\wt{\pi}_{\calL_0}(f)(1+|X|^2)^{-s/2}\right)  = 
      \Tr_\mathrm{Vol}(\wt{\pi}_\calL(f)) \,  \mathrm{Vol}_{d-1}(S^{d-1}).  \qedhere
\end{equation*}
\end{proof}

\ack 
EP acknowledges financial support from the W. M. Keck Foundation. CB is supported by a postdoctoral fellowship 
for overseas researchers from The Japan Society for the Promotion of Science (No. P16728)
and a KAKENHI Grant-in-Aid for JSPS fellows (No. 16F16728). 
Lemma \ref{lemma:d_dim_trace_per_unit_volume_related_to_residue_trace} is adapted 
from a similar result in the first author's thesis~\cite[Lemma 3.3.7]{BourneThesis}, the proof of which was done in collaboration 
with Adam Rennie.

\bibliographystyle{iopart-num}

\end{document}